\documentclass[journal,twocolumn]{IEEEtran}

\usepackage{amsmath}
\usepackage{amsfonts}
\usepackage{amssymb}
\usepackage[normalem]{ulem}
\usepackage{stackrel}
\usepackage{centernot}
\usepackage{cite}
\usepackage{graphicx}

\usepackage[
  colorlinks,
  linkcolor = blue,
  citecolor = blue,
  urlcolor = blue]{hyperref}

\usepackage{amsthm}
\usepackage[capitalize]{cleveref}

\newtheorem{theorem}{Theorem}
\newtheorem{lemma}[theorem]{Lemma}
\newtheorem{corollary}[theorem]{Corollary}

\crefname{conjecture}{conjecture}{conjectures}
\Crefname{conjecture}{Conjecture}{Conjectures}
\newtheorem{definition}[theorem]{Definition}

\Crefname{observation}{Observation}{Observations}

\crefname{secinapp}{Appendix}{Appendices}
\Crefname{secinapp}{Appendix}{Appendices}

\newcommand{\abs}[1]{\ensuremath{{\lvert#1\rvert}}}
\newcommand{\bra}[1]{\ensuremath{\left\langle{#1}\right|}}
\newcommand{\ket}[1]{\ensuremath{\left|{#1}\right\rangle}}
\newcommand{\braket}[2]{\ensuremath{\langle{#1}|{#2}\rangle}}
\newcommand{\braopket}[3]{\ensuremath{\left\langle{#1}\middle|{#2}\middle|{#3}\right\rangle}}
\newcommand{\ip}[1]{\langle{#1}\rangle}
\newcommand{\hilb}[1]{\ensuremath{\mathcal{H}_#1}}

\newcommand{\linspan}[0]{\ensuremath{\textnormal{span}}}

\newcommand{\Tr}[0]{\ensuremath{\textnormal{Tr}}}
\newcommand{\linop}[1]{\ensuremath{\mathcal{L}(#1)}}
\newcommand{\proj}[1]{{\ket{#1}\bra{#1}}}
\newcommand{\vecize}[0]{\textnormal{vec}}

\newcommand{\hA}[0]{\hilb{A}}
\newcommand{\hB}[0]{\hilb{B}}

\newcommand{\chan}[1]{\ensuremath{\mathcal{#1}}}

\newcommand{\ot}[0]{\otimes}
\newcommand{\Gc}[0]{{\overline{G}}}
\newcommand{\Hc}[0]{{\overline{H}}}

\newcommand{\thbar}[0]{{\bar{\vartheta}}}
\newcommand{\thp}[0]{{\vartheta^+}}
\newcommand{\thm}[0]{{\vartheta^-}}
\newcommand{\thpbar}[0]{{\bar{\vartheta}^+}}
\newcommand{\thmbar}[0]{{\bar{\vartheta}^-}}

\newcommand{\opnorm}[1]{\ensuremath{\left\lVert#1\right\rVert}}

\newcommand{\opV}[1]{\linop{\mathbb{C}^{\abs{V(#1)}}}}

\newcommand{\homm}[0]{\to}
\newcommand{\home}[0]{\stackrel{*}{\to}}
\newcommand{\homb}[0]{\stackrel{B}{\to}}
\newcommand{\homp}[0]{\stackrel{+}{\to}}
\newcommand{\homv}[0]{\stackrel{V}{\to}}
\newcommand{\homq}[0]{\stackrel{q}{\to}}

\newcommand{\chivect}[0]{\chi_{\textrm{vect}}}
\newcommand{\omegavect}[0]{\omega_{\textrm{vect}}}
\newcommand{\chiqr}[0]{\chi_{\textrm{qr}}}

\newcommand{\djp}[0]{*}

\newcommand{\simG}[0]{\sim_G}
\newcommand{\simH}[0]{\sim_H}

\newcommand{\cart}{\mathbin{\square}}

\begin{document}

\title{Bounds on Entanglement Assisted Source-channel Coding
    \texorpdfstring{via the Lov{\'a}sz $\vartheta$}{via the Lovasz Theta}
    Number and its Variants}

\author{Toby~Cubitt,
    Laura~Man\v{c}inska,
    David~Roberson,
    Simone~Severini,
    Dan~Stahlke,
    and~Andreas~Winter%
\thanks{Toby Cubitt is with the Mathematics and Quantum Information group, Departamento de
    An\'{a}lisis Matem\'{a}tico, Facultad de CC Matem\'{a}ticas, Universidad Complutense de
    Madrid, 28040 Madrid, Spain and the Department of Applied Mathematics and Theoretical
    Physics, University of Cambridge, Wilberforce Road, Cambridge, U.K. (e-mail:
    tsc25@cam.ac.uk).}
\thanks{Laura Man\v{c}inska is with the Institute for Quantum Computing, University of
    Waterloo, Canada and the Centre for Quantum Technologies, National University of Singapore,
    21 Lower Kent Ridge Road, Singapore (e-mail: laura.mancinska@gmail.com).}
\thanks{David Roberson is with the Department of Combinatorics \& Optimization, University of
    Waterloo, Canada and the Division of Mathematical Sciences, School of Physical and
    Mathematical Sciences, Nanyang Technological University, 50 Nanyang Avenue, Singapore
    (e-mail: davideroberson@gmail.com).}
\thanks{Simone Severini is with the Department of Computer Science, and Department of Physics
    and Astronomy, University College London, WC1E 6BT London, U.K. (e-mail:
    simoseve@gmail.com).}
\thanks{Dan Stahlke is with the Department of Physics, Carnegie Mellon University, Pittsburgh,
    Pennsylvania 15213, USA (e-mail: dan@stahlke.org)}
\thanks{Andreas Winter is with the ICREA \& F\'{\i}sica Te\`{o}rica: Informaci\'{o} i Fenomens
    Qu\`{a}ntics, Universitat Aut\`{o}noma de Barcelona, ES-08193 Bellaterra (Barcelona),
    Spain.
    When the core part of the present research was done he was affiliated with the
    Department of Mathematics, University of Bristol, BS8 1TW Bristol, U.K.
    (e-mail: andreas.winter@uab.cat)}
\thanks{Some of the results in the present paper were presented at the Theory of Quantum
    Computation, Communication \& Cryptography 2014 conference.}
}

%
%
%
%
%

\date{6 December 2013}

\maketitle

\begin{abstract}
    We study zero-error entanglement assisted source-channel coding (communication in the
    presence of side information).
    Adapting a technique of Beigi, we show that such coding requires existence of a set of vectors
    satisfying orthogonality conditions related to suitably defined graphs $G$ and $H$. Such
    vectors exist if and only if $\vartheta(\Gc) \le \vartheta(\Hc)$ where $\vartheta$
    represents the Lov{\'a}sz number. We also obtain similar inequalities for the related
    Schrijver $\thm$ and Szegedy $\thp$ numbers.

    These inequalities reproduce several known bounds and also lead to new results.  We provide
    a lower bound on the entanglement assisted cost rate.  We show that the entanglement
    assisted independence number is bounded by the Schrijver number: $\alpha^*(G) \le \thm(G)$.
    Therefore, we are able to disprove the conjecture that the one-shot entanglement-assisted
    zero-error capacity is equal to the integer part of the Lov{\'a}sz number.  Beigi
    introduced a quantity $\beta$ as an upper bound on $\alpha^*$ and posed the question of
    whether $\beta(G) = \lfloor \vartheta(G) \rfloor$. We answer this in the affirmative and
    show that a related quantity is equal to $\lceil \vartheta(G) \rceil$.  We show that a
    quantity $\chivect(G)$ recently introduced in the context of Tsirelson's problem is
    equal to $\lceil \vartheta^+(\Gc) \rceil$.

    In an appendix we investigate multiplicativity properties of Schrijver's and Szegedy's
    numbers, as well as projective rank.
\end{abstract}


\IEEEpeerreviewmaketitle

\section{Introduction}

The source-channel coding problem is as follows: Alice and Bob can communicate only through a
noisy channel.  Alice wishes to send a message to Bob, and Bob already has some side
information regarding Alice's message.
(Note that Alice's message may be several bits long.)
Alice encodes her message and sends a transmission through the channel.
Given the (noisy) channel output along with his side information, Bob must be able to deduce
Alice's message with zero probability of error (we always require zero error
throughout this entire paper).
An \emph{$(m,n)$-coding scheme} consists of encoding and decoding operations which allow
sending $m$ messages via $n$ uses of the noisy channel (again, each of the $m$ messages
may be several bits long).
The \emph{cost rate} $\eta$ is the infimum of $n/m$ over all $(m,n)$-coding schemes.

There are two special cases which are particularly noteworthy.
If the messages are bits and there is no side information then the inverse of the cost rate,
$1/\eta$, is the
\emph{Shannon capacity}~\cite{shannon56}, the number of zero-error bits that can be transmitted
per channel use in the limit of many uses of the channel.
On the other hand,
communication over a perfect channel with side information was considered by
Witsenhausen~\cite{witsen76}; the corresponding cost rate is known as the \emph{Witsenhausen rate}.
The general problem, with both side information and a noisy channel, was considered by
Nayak, Tuncel, and Rose~\cite{1705019}.

The Shannon capacity of a channel is very difficult to compute, and is
not even known to be decidable.
However, a useful upper bound on Shannon capacity is provided by the $\vartheta$ number
introduced by Lov{\'a}sz~\cite{lovasz79}.
The Lov{\'a}sz $\vartheta$ number also provides a lower bound on the Witsenhausen
rate~\cite{1705019} and, in general, the cost rate.

Recently it has been of interest to study a version of this problem in which the parties may
make use of an entangled quantum state, which can in certain cases increase the zero-error
capacity of a classical channel~\cite{PhysRevLett.104.230503,leung2012entanglement}.
The Lov{\'a}sz $\vartheta$ number upper bounds entanglement assisted
Shannon capacity, just as it does classical Shannon
capacity~\cite{PhysRevA.82.010303,dsw2013}.
Beigi's proof~\cite{PhysRevA.82.010303} proceeds through a relaxation of the channel coding
problem, with the relaxed
constraints consisting of various orthogonality conditions imposed upon a set of vectors.
We study a relaxation of the entanglement assisted source-channel coding problem inspired
by this technique of Beigi.
This relaxation leads to a set of constraints that are exactly characterized by
monotonicity of $\vartheta$.
This has a number of consequences.
Beigi defined a function $\beta$ as an upper bound on entanglement assisted independence
number and posed the question of whether $\beta$ is equal to
$\lfloor \vartheta \rfloor$.
We answer this in the affirmative and show that a similarly defined quantity is equal to
$\lceil \vartheta \rceil$.
We show that $\vartheta$ provides a bound for the source-channel coding problem.
As a special case this reproduces both Beigi's result as well as that of Bri\"{e}t et
al.~\cite{arxiv:1308.4283} in which it is shown that $\vartheta$ is a lower bound on the
entanglement assisted Witsenhausen rate.

A slightly different relaxation of source-channel coding leads to
three necessary conditions for the existence of a $(1,1)$-coding scheme in terms of
$\vartheta$ and two variants: Schrijver's $\thm$ and Szegedy's $\thp$.
This reproduces or strengthens results
from~\cite{PhysRevA.82.010303,arxiv:1308.4283,arxiv:1212.1724}
under a unified framework, with simpler proofs.
In particular, we produce a tighter bound on the entanglement assisted independence number:
$\alpha^* \le \thm$.

The technical results, \cref{thm:th_homb,thm:th_homp}, should be accessible to the reader who is
familiar with graph theory but not information theory or quantum mechanics, which merely
provide a motivation for the problem.

\section{Source-channel coding}

We will make use of the following graph theoretical concepts.
A \emph{graph} $G$ consists of a set of \emph{vertices} $V(G)$ along with a symmetric binary
relation $x \simG y$ among vertices (we abbreviate $x \sim y$ when the graph can be
inferred from context).
A pair of vertices $(x,y)$ satisfying $x \sim y$ are said to be \emph{adjacent}.
Equivalently, it is said that there is an \emph{edge} between $x$ and $y$.
Vertices are not adjacent to themselves, so $x \not\sim x$ for all $x \in V(G)$.
The \emph{complement} of a graph $G$, denoted $\Gc$, has the same set of vertices but has edges
between distinct pairs of vertices which are not adjacent in $G$
(i.e.\ for $x \ne y$ we have $x \sim_{\Gc} y \iff x \not\simG y$).
A set of vertices no two of which form an edge is known as an \emph{independent set}; the size
of the largest independent set is the \emph{independence number} $\alpha(G)$.
A set of vertices such that every pair is adjacent is known as a \emph{clique}; the size of
the largest clique is the \emph{clique number} $\omega(G)$.  Clearly $\omega(G)=\alpha(\Gc)$.
An assignment of colors to vertices such that adjacent vertices are given distinct colors is
called a \emph{proper coloring}; the minimum number of colors needed is the \emph{chromatic
number} $\chi(G)$.
A function mapping the vertices of one graph to those of another, $f : V(G) \to V(H)$, is a
\emph{homomorphism} if $x \simG y \implies f(x) \simH f(y)$.
Since vertices are not adjacent to themselves it is necessary that $f(x) \ne f(y)$ when
$x \sim y$.
If such a function exists, we say that $G$ \emph{is homomorphic to} $H$ and write
$G \homm H$.
The \emph{complete graph} on $n$ vertices, denoted $K_n$, has an edge between
every pair of vertices.
It is not hard to see that $\omega(G)$ is equal to the largest $n$ such that $K_n \homm G$,
and $\chi(G)$ is equal to the smallest $n$ such that $G \homm K_n$.
Many other graph properties can be expressed in terms of homomorphisms; for details
see~\cite{hahn1997graph,HellNesetril200409}.
The \emph{strong product} of two graphs, $G \boxtimes H$, has vertex set $V(G) \times V(H)$ and
has edges
\begin{align*}
    (x_1,y_1) \sim (x_2,y_2) \iff
    &(x_1=x_2 \textrm{ and } y_1 \sim y_2) \textrm{ or }
    \\ &(x_1 \sim x_2 \textrm{ and } y_1 = y_2) \textrm{ or }
    \\ &(x_1 \sim x_2 \textrm{ and } y_1 \sim y_2).
\end{align*}
The $n$-fold strong product is written
$G^{\boxtimes n} := G \boxtimes G \boxtimes \dots \boxtimes G$.
The \textit{disjunctive product} $G \djp H $ has edges
\begin{align*}
    (x_1,y_1) \sim (x_2,y_2) \iff
    x_1 \sim x_2 \textrm{ or } y_1 \sim y_2.
\end{align*}
It is easy to see that $\overline{G \djp H} = \Gc \boxtimes \Hc$.
The $n$-fold disjunctive product is written
$G^{\djp n} := G \djp G \djp \dots \djp G$.

Suppose that Alice communicates to Bob through a noisy classical channel
$\chan{N} : S \to V$.
She wishes to send a message to Bob with zero chance of error.
Let $\chan{N}(v|s)$ denote the probability that $\chan{N}$ will produce symbol $v$
when given symbol $s$ as input, and define the graph $H$ with vertices $S$ and edges
\begin{align}
    \label{eq:confgraph_chap4}
    s \simH t \iff \chan{N}(v|s) \chan{N}(v|t) = 0 \textrm{ for all } v \in V.
\end{align}
Bob can distinguish codewords $s$ and $t$ if and only if they have no chance of being mapped
to the same output by $\chan{N}$.
Therefore, Alice's set of codewords must form a clique of $H$; the size of the
largest such set is the clique number $\omega(H)$.
We will call this the \emph{distinguishability graph} of $\chan{N}$.
The complementary graph $\Hc$ is known as the \emph{confusability graph} of $\chan{N}$.
Note that standard convention is to denote the confusability graph by $H$ rather than $\Hc$.
We break convention in order to make notation in this paper much simpler.
However, to minimize confusion when discussing prior results,
we will follow the tradition of using the independence number when speaking of the number of
codewords that Alice can send (equal to $\alpha(\Hc) = \omega(H)$ in our notation).

The number of bits (the base-2 log of the number of distinct codewords)
that Alice can send with a single use of $\chan{N}$ is known as the
\emph{one-shot zero-error capacity} of $\chan{N}$, and is equal to $\log \alpha(\Hc)$.
The average number of bits that can be sent per channel use (again with zero error) in the limit
of many uses of a channel is known as the \emph{Shannon capacity}.
With $n$ parallel uses of $\chan{N}$ the distinguishability graph is
$H^{\djp n}$.
The Shannon capacity of $\chan{N}$ is therefore
\begin{align*}
    \Theta(\Hc) :&= \lim_{n \to \infty} \frac{1}{n} \log \omega(H^{\djp n})
    \\ &= \lim_{n \to \infty} \frac{1}{n} \log \alpha(\Hc^{\boxtimes n}).
\end{align*}

This quantity is in general very difficult to compute, with the capacity of the five cycle
graph $\Hc=C_5$ having been open for over 20 years and the capacity of $C_7$ being
unknown to this day.
The capacity of $C_5$ was solved by Lov{\'a}sz~\cite{lovasz79} who introduced a function
$\vartheta(\Hc)$, the definition of which will be postponed until \cref{sec:mainthms}.
Lov{\'a}sz proved a sandwich theorem which, using the notation $\thbar(H) := \vartheta(\Hc)$,
takes the form
\begin{align*}
    \alpha(\Hc) = \omega(H) \le \thbar(H) \le \chi(H).
\end{align*}
He also showed that $\thbar(H^{\djp n}) = \thbar(H)^n$, therefore
$\Theta(\Hc) \le \log \thbar(H)$.
This bound also applies to entanglement assisted
communication~\cite{PhysRevA.82.010303}, which we will investigate in detail, and has been
generalized to quantum channels~\cite{dsw2013}.

\begin{figure}[ht]
    \centering
    \includegraphics[scale=1.0]{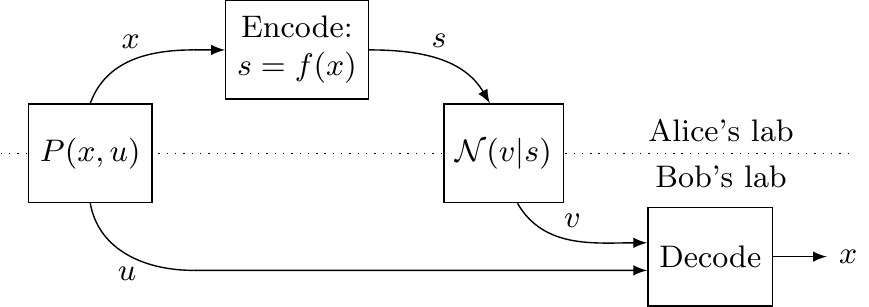}
    \caption{
        A zero-error source-channel $(1,1)$-coding scheme.
    }
    \label{fig:coding_chap4}
\end{figure}

We now introduce the \emph{source-channel coding} problem.
As before, Alice wishes to send Bob a message $x \in X$, and she can only communicate through a
noisy channel $\chan{N} : S \to V$.
Now, however, Bob has some side information about Alice's message.
Specifically, Alice and Bob each receive one part of a pair $(x,u)$ drawn according to a
probability distribution $P(x,u)$.  This is known as a \emph{dual source}.
Alice encodes her input $x$ using a function $f : X \to S$ and sends the result through $\chan{N}$.
Bob must deduce $x$ with zero chance of error using the output of $\chan{N}$ along with his
side information $u$.
Such a protocol is called a \emph{zero-error source-channel $(1,1)$-coding scheme}, and is
depicted in \cref{fig:coding_chap4}.
An $(m,n)$-coding scheme transmits $m$ independent instances of the source using $n$ copies of
the channel.

Again the analysis involves graphs.
Let $H$ again be the distinguishability graph~\eqref{eq:confgraph_chap4} and define
the \emph{characteristic graph} $G$ with vertices $X$ and edges
\begin{align*}
    x \simG y \iff \exists u \in U \textrm{ such that } P(x,u)P(y,u) \ne 0.
\end{align*}
In~\cite{1705019} it was shown that decoding is possible if and only if Alice's encoding $f$ is
a homomorphism from $G$ to $H$.\footnote{
Basically, $G$ represents the information that needs to be sent and $H$ represents the
information that survives the channel.  A homomorphism $G \homm H$ ensures that the needed
information makes it through the channel intact.
}
Therefore a zero-error $(1,1)$-coding scheme exists if and only if $G \homm H$.
A zero-error $(m,n)$-coding scheme is possible if and only if
\begin{align}
    \label{eq:mn_hom}
    G^{\boxtimes m} \homm H^{\djp n}.
\end{align}
The smallest possible ratio $n/m$ (in the limit $m \to \infty$) is called the
\emph{cost rate}, $\eta(G, \Hc)$.
More precisely, the cost rate is defined as
\begin{align}
    \label{eq:costrate}
    \eta(G, \Hc) = \lim_{m \to \infty} \frac{1}{m} \min\left\{
        n : G^{\boxtimes m} \homm H^{\djp n}
        \right\}.
\end{align}

The $\thbar$ quantity is monotone under graph homomorphisms in the sense that
$G \homm H \implies \thbar(G) \le \thbar(H)$~\cite{de2013optimization}.
Consequently, a zero-error $(1,1)$-coding scheme requires
$\thbar(G) \le \thbar(H)$.
Since $\thbar(G^{\boxtimes m}) = \thbar(G)^m$~\cite{knuth94} and
$\thbar(H^{\djp n}) = \thbar(H)^n$~\cite{lovasz79}, it follows that an
$(m,n)$-coding scheme is possible only if
$\log \thbar(G) / \log \thbar(H) \le n/m$.
Thus we have the bound
\begin{align*}
    \eta(G, \Hc) \ge \frac{\log \thbar(G)}{\log \thbar(H)}.
\end{align*}
(Cf.~\cite{1705019} for the special case of the Witsenhausen rate.)

We will return to this in \cref{sec:mainthms} when we prove an analogous bound for
entanglement assisted zero-error source-channel coding.

When Bob has no side information (equivalently, when $U$ is a singleton), $G$ is the complete
graph.  In this case zero-error transmission of $x$ is possible if and only if $K_n \homm H$
where $n=\abs{X}$, which in turn holds if and only if $n \le \omega(H) = \alpha(\Hc)$.
This is the expected result, since as mentioned before $\alpha(\Hc)$ is the number of
unambiguously decodable codewords that Alice can send through $\chan{N}$.
On the other hand, consider the case where there is side information and
$\chan{N}$ is a noiseless channel of size $n=\abs{S}$.
Now $H$ becomes the complete graph $K_n$, so $x$ can be perfectly transmitted if and only if
$G \homm K_n$.
This holds if and only if $n \ge \chi(G)$.
These two examples provide an operational interpretation to the independence number and
chromatic number of a graph.  The analogous communication problems in the presence of an
entangled state (which we examine shortly) define the entanglement assisted independence
and chromatic numbers.

\begin{figure}[ht]
    \centering
    \includegraphics[scale=1.0]{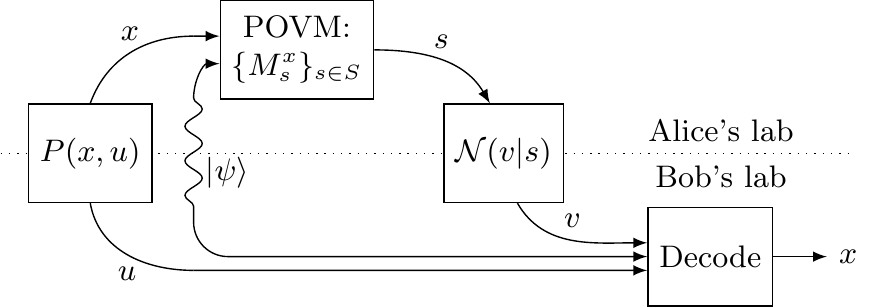}
    \caption{
        An entanglement assisted zero-error source-channel $(1,1)$-coding scheme.
    }
    \label{fig:coding_ent}
\end{figure}

If Alice and Bob share an entangled state they can use the strategy depicted in
\cref{fig:coding_ent}, which is described in greater detail in~\cite{arxiv:1308.4283}.
Alice, upon receiving $x \in X$, performs a POVM $\{ M^x_s \}_{s \in S}$ on her
half of the entanglement resource $\ket{\psi} \in \hA \ot \hB$ and receives measurement
outcome $s \in S$.
Without loss of generality this can be assumed to be a projective measurement since any POVM
can be converted to a projective measurement by enlarging the entangled state.
So for each $x \in X$, the collection $\{ M^x_s \}_{s \in S}$ consists of projectors on $\hA$
which sum to the identity.
Alice sends the measurement outcome $s$ through the channel $\chan{N}$ to Bob, who receives
some $v \in V$ such that $\chan{N}(v|s) > 0$.
Bob then measures his half of the entangled state using a projective measurement depending on
$v$ and his side information $u$.
An \emph{entanglement assisted zero-error $(1,1)$-coding scheme}
is one in which Bob is able to determine $x$ with
zero chance of error; an
\emph{entanglement assisted zero-error $(m,n)$-coding scheme}
involves sending $m$ independent samples of the source using $n$ copies of the channel.

After Alice's measurement, Bob's half of the entanglement resource is in the state
\begin{align*}
    \rho_s^x = \Tr_A\{ (M_s^x \ot I) \proj{\psi} \}.
\end{align*}
An error free decoding operation exists for Bob if and only if these states are orthogonal
for every $x \in X$ consistent with the information in Bob's possession (i.e. $u$ and $v$).
We then have the following necessary and sufficient condition~\cite{arxiv:1308.4283}.
Let $G$ be the characteristic graph of the source and $H$ be the
distinguishability graph of the channel.
There must be a bipartite pure state $\ket{\psi} \in \hA \ot \hB$ for some Hilbert spaces $\hA$
and $\hB$, and for each $x \in X$ there must be a projective decomposition of the identity
$\{ M^x_s \}_{s \in S}$ on $\hA$ such that
\begin{align*}
    \rho_s^x \perp \rho_t^y
    \textnormal{ for all } x \simG y \textnormal{ and } s \not\simH t,
\end{align*}
where orthogonality is in terms of the Hilbert--Schmidt inner product.

Recall that without entanglement a zero-error $(1,1)$-coding scheme was possible if and only
if $G \homm H$.
By analogy we say there is an \emph{entanglement assisted homomorphism} $G \home H$
when there exists an entanglement assisted zero-error $(1,1)$-coding scheme:
\begin{definition}
    \label{def:home}
    Let $G$ and $H$ be graphs.  There is an entanglement assisted homomorphism from $G$ to $H$,
    written $G \home H$, if there is a bipartite state $\ket{\psi} \in \hA \ot \hB$ (for some
    Hilbert spaces $\hA$ and $\hB$) and, for each $x \in V(G)$, a projective decomposition of
    the identity $\{ M_s^x \}_{s \in V(H)}$ on $\hA$ such that
    \begin{align}
        \label{eq:cond_home}
        \rho_s^x \perp \rho_t^y
        \textnormal{ for all } x \simG y \textnormal{ and } s \not\simH t,
    \end{align}
    where
    \begin{align}
        \label{eq:rho_sx}
        \rho_s^x := \Tr_A\{ (M_s^x \ot I) \proj{\psi} \}.
    \end{align}
\end{definition}

Analogous to~\eqref{eq:mn_hom}, there is an entanglement assisted $(m,n)$-coding
scheme if and only if $G^{\boxtimes m} \home H^{\djp n}$.
The entangled cost rate~\cite{arxiv:1308.4283} is analogous to~\eqref{eq:costrate},
\begin{align}
    \label{eq:entcostrate}
    \eta^*(G, \Hc) = \lim_{m \to \infty} \frac{1}{m} \min\left\{
        n : G^{\boxtimes m} \home H^{\djp n}
        \right\}.
\end{align}
In the absence of side information (i.e.\ with $U$ being a singleton set),
$G$ becomes the complete graph.
We saw above that without entanglement and without side information,
$n$ distinct codewords can be sent error-free through a noisy channel if and only if
$K_n \homm H$; the largest such $n$ is $\omega(H) = \alpha(\Hc)$.
With the help of entanglement the largest number of codewords
is the largest $n$ such that $K_n \home H$; this defines the
\emph{entanglement assisted independence number}, $\alpha^*(\Hc)$.
Since an entanglement resource never hurts, $\alpha^*(\Hc) \ge \alpha(\Hc)$ always.
In some cases $\alpha^*(\Hc)$ can be strictly larger than
$\alpha(\Hc)$~\cite{PhysRevLett.104.230503}.

We saw above that $\alpha(\Hc) \le \thbar(H)$.  Indeed, this was the original application of
$\thbar$.
Beigi showed that also $\alpha^*(\Hc) \le \thbar(H)$~\cite{PhysRevA.82.010303}
(this has been generalized to quantum channels as
well~\cite{dsw2013}; however, we consider here only classical channels).
Beigi proved his bound by showing that if $n$ distinct codewords can be sent through a noisy
channel with zero-error using entanglement ($K_n \home H$ in our notation) then there are
vectors $\ket{w} \ne 0$ and $\ket{w_s^x}$ with $x \in \{1,\dots,n\}$ and $s \in V(H)$ such
that\footnote{
    Recall that we take $\Hc$ to be the confusability graph rather than $H$.  So Beigi's
    definition is worded differently.
}
\begin{align}
       &\; \sum_s \ket{w_s^x} = \ket{w} \label{eq:beigi_first}
    \\ &\; \braket{w_s^x}{w_t^y} = 0 \textnormal{ for all } x \ne y, s \not\simH t
    \\ &\; \braket{w_s^x}{w_t^x} = 0 \textnormal{ for all } s \ne t. \label{eq:beigi_last}
\end{align}
Denote by $\beta(\Hc)$ the largest $n$ such that vectors of this form exist.
Then $\beta(\Hc) \ge \alpha^*(\Hc)$.
Beigi showed that the existence of such vectors implies $n \le \thbar(H)$,
therefore $\alpha^*(\Hc) \le \beta(\Hc) \le \thbar(H)$.
Since $\vartheta$ is multiplicative under the strong graph product, $\thbar(H)$ is in fact
an upper bound on the entanglement assisted Shannon capacity.
Beigi left open the question of whether $\beta(\Hc)$ was equal to $\lfloor \thbar(H) \rfloor$.
We will answer this question in the affirmative (\cref{thm:beigi_beta}).

In fact, we show something more general.
We generalize Beigi's vectors so that they apply to the source-channel coding problem (i.e.\
with $G$ not necessarily being $K_n$) and give a bound in terms of the Lov{\'a}sz $\vartheta$
number.
The conditions we will introduce can be thought of as a relaxation of the
condition~\eqref{eq:cond_home}, which defines $G \home H$.
A related but different relaxation will give bounds in terms of two variations of the
Lov{\'a}sz number: the Schrijver number~\cite{1056072,mceliece1978lovasz} and the Szegedy
number~\cite{365707}.
We denote the first relaxation $G \homb H$ since it generalizes Beigi's condition, and denote
the second $G \homp H$ since it contains a positivity condition.
A third relaxation, $G \homv H$, is defined here but the significance is discussed later.

\begin{definition}
    \label{def:hombpv}
    Let $G$ and $H$ be graphs.
    Write $G \homb H$ if there are vectors $\ket{w} \ne 0$ and $\ket{w_s^x} \in \mathbb{C}^d$ for
    each $x \in V(G)$, $s \in V(H)$, for some $d \in \mathbb{N}$, such that
    \begin{enumerate}
        \item $\sum_s \ket{w_s^x} = \ket{w}$
        \item $\braket{w_s^x}{w_t^y} = 0$ for all $x \simG y$, $s \not\simH t$
        \item $\braket{w_s^x}{w_t^x} = 0$ for all $s \ne t$.
    \end{enumerate}
    Write $G \homp H$ if there are vectors satisfying
    \begin{enumerate}
        \item $\sum_s \ket{w_s^x} = \ket{w}$
        \item $\braket{w_s^x}{w_t^y} = 0$ for all $x \simG y$, $s \not\simH t$
        \item $\braket{w_s^x}{w_t^y} \ge 0$.
    \end{enumerate}
    Write $G \homv H$ if there are vectors satisfying the conditions for both
    $G \homb H$ and $G \homp H$, i.e.,
    \begin{enumerate}
        \item $\sum_s \ket{w_s^x} = \ket{w}$
        \item $\braket{w_s^x}{w_t^y} = 0$ for all $x \simG y$, $s \not\simH t$
        \item $\braket{w_s^x}{w_t^x} = 0$ for all $s \ne t$
        \item $\braket{w_s^x}{w_t^y} \ge 0$.
    \end{enumerate}
\end{definition}

Without loss of generality one could consider only real vectors, since complex
vectors can be turned real via the recipe
$\ket{\hat{w}_s^x} = \textrm{Re}(\ket{w_s^x}) \oplus \textrm{Im}(\ket{w_s^x})$
while preserving the inner product properties required by the above definitions.

It is enlightening to consider the Gram matrices of the $\ket{w_s^x}$ vectors.
In fact, it is this formulation that will be used to prove our main theorems.

\begin{theorem}
    \label{thm:gram}
    $G \homb H$ if and only if there is a positive semidefinite matrix
    $C : \opV{G} \ot \opV{H}$ satisfying
    \begin{align}
        &\, \sum_{s,t} C_{xyst} = 1 \label{eq:Cb_first}
        \\ &\, C_{xyst} = 0 \textrm{ for } x \simG y \textrm{ and } s \not\simH t
        \label{eq:Cb_second}
        \\ &\, C_{xxst} = 0 \textrm{ for } s \ne t. \label{eq:Cb_last}
    \end{align}
    $G \homp H$ if and only if there is a positive semidefinite matrix
    satisfying~\eqref{eq:Cb_first},~\eqref{eq:Cb_second}, and
    \begin{align}
        C_{xyst} \ge 0. \label{eq:Cp}
    \end{align}
    $G \homv H$ if and only if there is a positive semidefinite matrix
    satisfying~\eqref{eq:Cb_first}-\eqref{eq:Cp}.
\end{theorem}
\begin{proof}
    We prove only the $G \homb H$ claim; the proofs for $G \homp H$ and $G \homv H$ are analogous.
    Suppose $G \homb H$ and let $\ket{w}$ and $\ket{w_s^x}$ be the vectors described in
    \cref{def:hombpv}.
    Without loss of generality rescale so that $\braket{w}{w} = 1$.
    Define the matrix $C : \opV{G} \ot \opV{H}$ with entries
    $C_{xyst} = \braket{w_s^x}{w_t^y}$ for $x,y \in V(G)$ and $s,t \in V(H)$.
    Since $C$ is a Gram matrix, it is positive semidefinite.
    Properties~\eqref{eq:Cb_first}-\eqref{eq:Cb_last} follow directly from the three properties
    listed in \cref{def:hombpv} for $G \homb H$ (the first of these uses $\braket{w}{w} = 1$).

    For the converse, note that any positive semidefinite matrix $C : \opV{G} \ot \opV{H}$
    is a Gram matrix of some vectors $\ket{w_s^x}$.
    The three properties of \cref{def:hombpv} for $G \homb H$
    follow from~\eqref{eq:Cb_first}-\eqref{eq:Cb_last}.
    Only the first of these is nontrivial.  We have that for all $x,y \in V(G)$,
    \begin{align*}
        1 = \sum_{st} C_{xyst} = \sum_{st} \braket{w_s^x}{w_t^y} =
        \left( \sum_s \bra{w_s^x} \right)\left( \sum_t \ket{w_t^y} \right).
    \end{align*}
    For $x = y$, the above implies that $\sum_s \ket{w_s^x}$ is a unit vector for all $x$.
    These unit vectors must have unit inner product amongst themselves by the $x \ne y$
    cases above, and therefore they must all be the same vector.
    Call this $\ket{w}$.  Clearly $\ket{w} \ne 0$.
\end{proof}

It is interesting to note that a matrix with
properties~\eqref{eq:Cb_first},~\eqref{eq:Cb_second}, and~\eqref{eq:Cp} (those associated
with $G \homp H$) can be interpreted as
a conditional probability distribution, $P(s,t|x,y) = C_{xyst}$.
With this interpretation, $G \homp H$ if and only if there exists a conditional
probability distribution $P(s,t|x,y)$ such that $C_{sx;ty}=P(s,t|x,y)$ is a positive
semidefinite matrix and $P(s \simH t|x \simG y)=1$.
A similar interpretation holds for $G \homv H$.

We now show that $G \homb H$ and $G \homp H$ are indeed relaxations of $G \home H$
(the significance of $G \homv H$ will be explained later).
Since \cref{def:hombpv} reduces to Beigi's criteria when considering $K_n \homb H$, the
argument that follows provides an alternative and simpler proof of Beigi's result that
$\alpha^*(\Hc) \le \beta(\Hc)$.

\begin{theorem}
    \label{thm:hom_relax}
    If $G \home H$ then $G \homb H$ and $G \homp H$.
\end{theorem}
\begin{proof}
    ($G \home H \implies G \homb H$):
    Suppose that $G \home H$.  Let $\ket{\psi}$ and $M_s^x$ for $x \in V(G)$
    and $s \in V(H)$ satisfy condition~\eqref{eq:cond_home} (with $\rho_s^x$ given
    by~\eqref{eq:rho_sx}).
    Define $\ket{w} = \ket{\psi}$ and
    \begin{align*}
        \ket{w_s^x} = (M_s^x \ot I) \ket{\psi}.
    \end{align*}
    Since $\{ M^x_s \}_{s \in S}$ is a projective decomposition of the identity,
    \begin{align*}
        \sum_s M_s^x = I &\implies \sum_s \ket{w_s^x} = \ket{w},
        \\
        M_s^x M_t^x = 0 &\implies \braket{w_s^x}{w_t^x} = 0 \textnormal{ for } s \ne t.
    \end{align*}
    For all $x \simG y$ and $s \not\simH t$, condition~\eqref{eq:cond_home} gives
    that the reduced density operators (tracing over $\hA$) of the post-measurement states
    $(M_s^x \ot I) \ket{\psi}$ and $(M_t^y \ot I) \ket{\psi}$ are orthogonal.  But this is only
    possible if the pure states (without tracing out $\hA$) are orthogonal.
    So,
    \begin{align*}
        \braopket{\psi}{(M_s^x \ot I)^\dag (M_t^y \ot I)}{\psi} = 0
        &\implies
        \braket{w_s^x}{w_t^y} = 0.
    \end{align*}

    ($G \home H \implies G \homp H$):
    Suppose that $G \home H$.  Let $\ket{\psi}$ and $M_s^x$ for $x \in V(G)$
    and $s \in V(H)$ satisfy condition~\eqref{eq:cond_home} (with $\rho_s^x$ given
    by~\eqref{eq:rho_sx}).
    Define $\ket{w_s^x}$ to be the vectorization of the post-measurement reduced density
    operator,
    \begin{align*}
        \ket{w_s^x} = \vecize(\rho_s^x).
    \end{align*}
    Since $\{ M^x_s \}_{s \in S}$ sum to identity,
    \begin{align*}
        \sum_s \ket{w_s^x} &=
        \vecize\left( \sum_s \Tr_A\{ (M_s^x \ot I) \proj{\psi} \} \right)
        \\ &= \vecize(\Tr_A\{ \proj{\psi} \}) =: \ket{w}.
    \end{align*}
    For all $x \simG y$ and $s \not\simH t$, condition~\eqref{eq:cond_home} gives
    $\braket{w_s^x}{w_t^y} = 0$.
    Density operators are positive, giving positive inner products
    $\braket{w_s^x}{w_t^y} \ge 0$.
\end{proof}

\section{Monotonicity theorems}
\label{sec:mainthms}

Our main results concern monotonicity properties of the Lov{\'a}sz number $\vartheta$,
Schrijver number $\thm$, and Szegedy number $\thp$ for graphs that are related by the
generalized homomorphisms of \cref{def:hombpv}.
These will lead to various bounds relevant to entanglement assisted zero-error source-channel
coding.
These three quantities are defined as follows.

\begin{definition}
    \label{def:theta}
    In this definition we use real matrices.
    For convenience we state the definitions in terms of the
    complement of a graph, since this form is used throughout the theorems.

    The Lov{\'a}sz number of the complement, $\thbar(G):=\vartheta(\Gc)$,
    is given by either of the following
    two semidefinite programs, which are equivalent~\cite{lovasz79,knuth94,lovaszsemidef}:
    \begin{align}
        \thbar(G) &= \max\{ \opnorm{I + T} : I + T \succeq 0,
            \notag \\ &\hphantom{= \max\{\;}
            T_{ij}=0 \textnormal{ for } i \not\sim j \},
            \label{eq:th_max_sdp}
        \\ \thbar(G) &= \min\{ \lambda : \exists Z \succeq 0, Z_{ii} = \lambda-1,
            \notag \\ &\hphantom{= \min\{\;}
            Z_{ij} = -1 \textnormal{ for } i \sim j \},
            \label{eq:th_min_sdp}
        \\ \intertext{where $\opnorm{\cdot}$ denotes the operator norm (the largest singular
        value) and $\succeq 0$ means that a matrix is positive semidefinite.
        The Schrijver number of the complement, $\thmbar(G) := \thm(\Gc)$,
        (sometimes written $\vartheta'$) is~\cite{1056072,mceliece1978lovasz}
        }
        \thmbar(G) &= \min\{ \lambda : \exists Z \succeq 0, Z_{ii} = \lambda-1,
            \notag \\ &\hphantom{= \min\{\;}
            Z_{ij} \le -1 \textnormal{ for } i \sim j \}.
            \label{eq:thm_min_sdp}
        \\ \intertext{The Szegedy number of the complement, $\thpbar(G) := \thp(\Gc)$,
        is~\cite{365707}
        }
        \thpbar(G) &= \min\{ \lambda : \exists Z \succeq 0, Z_{ii} = \lambda-1,
            \notag \\ &\hphantom{= \min\{\;}
            Z_{ij} = -1 \textnormal{ for } i \sim j,
            \notag \\ &\hphantom{= \min\{\;}
            Z_{ij} \ge -1 \textnormal{ for all } i,j \}.
            \label{eq:thp_min_sdp}
    \end{align}
    Clearly $\thmbar(G) \le \thbar(G) \le \thpbar(G)$.
\end{definition}

Our first result is that $G \homb H$ exactly characterizes ordering of $\thbar$.
This will lead to a bound on entanglement assisted cost rate.

\begin{theorem}
    \label{thm:th_homb}
    $G \homb H \iff \thbar(G) \le \thbar(H)$.
\end{theorem}
\begin{proof}
    ($\Longleftarrow$):
    Suppose $\thbar(G) \le \thbar(H)$.  We will explicitly construct a matrix
    $C \succeq 0$ satisfying properties~\eqref{eq:Cb_first}-\eqref{eq:Cb_last} of
    \cref{thm:gram}.
    Let $\lambda = \thbar(H)$.
    By definition, there is a matrix $T$ such that $\opnorm{I+T}=\lambda$, $I+T \succeq 0$, and
    $T_{st}=0$ for $s \not\sim t$.
    With $\ket{\psi}$ denoting the vector corresponding to the largest eigenvalue of $I+T$, and
    with $\circ$ denoting the Schur--Hadamard (i.e.\ entrywise) product, define the matrices
    \begin{align*}
        D &= \proj{\psi} \circ I,
        \\ B &= \proj{\psi} \circ (I + T).
    \end{align*}
    With $J$ being the all-ones matrix and $\ip{\cdot,\cdot}$ denoting the Hilbert--Schmidt inner
    product, it is readily verified that
    \begin{align*}
        \ip{D, J} &= \braket{\psi}{\psi} = 1, \\
        \ip{B, J} &= \braopket{\psi}{I+T}{\psi} = \lambda.
    \end{align*}
    Since the Schur--Hadamard product of two matrices is a principal submatrix of their tensor product, this operation preserves positive semidefiniteness.
    As a consequence, $B \succeq 0$ and
    \begin{align*}
        \opnorm{I+T} = \lambda &\implies \lambda I - (I+T) \succeq 0 \implies \lambda D - B
        \succeq 0.
    \end{align*}

    Since $\lambda \ge \thbar(G)$, there is a matrix $Z$ such that
    $Z \succeq 0$,
    $Z_{xx} = \lambda-1$ for all $x$, and
    $Z_{xy} = -1$ for all $x \simG y$.
    Note that \cref{def:theta} gives existence of a matrix with $\thbar(G)-1$ on the diagonal, but
    since $\lambda \ge \thbar(G)$ we can add a multiple of the identity to get $\lambda-1$ on the
    diagonal.

    We now construct $C$.  Define
    \begin{align*}
        C = \lambda^{-1} \left[ J \ot B + (\lambda-1)^{-1} Z \ot (\lambda D - B) \right].
    \end{align*}
    Since $J$, $B$, $Z$, and $\lambda D - B$ are all positive semidefinite, and $\lambda-1 \ge 0$,
    we have that $C$ is positive semidefinite.
    The other desired conditions on $C$ are easy to verify.
    For all $x,y$ we have
    \begin{align*}
        \sum_{st} C_{xyst}
        &= \lambda^{-1} \left[ \ip{B,J} +
            (\lambda-1)^{-1} Z_{xy} [\lambda\ip{D, J}-\ip{B, J}] \right]
        \\ &= 1.
    \end{align*}
    Note that the $J$ in the above equation is indexed by $V(H)$, whereas the $J$ in the definition of $C$ is indexed by $V(G)$. For $x \simG y$ and $s \not\simH t$,
    \begin{align*}
        C_{xyst} &= \lambda^{-1} \left[ B_{st} + (\lambda-1)^{-1} Z_{xy} (\lambda D_{st} -
            B_{st}) \right]
        \\ &= \lambda^{-1} \left[ B_{st} + (\lambda-1)^{-1} (-1) (\lambda B_{st} - B_{st}) \right]
        = 0.
    \end{align*}
    For all $x$ and for $s \ne t$,
    \begin{align*}
        C_{xxst} &= \lambda^{-1} \left[ B_{st} + (\lambda-1)^{-1} Z_{xx} (\lambda D_{st} -
            B_{st}) \right]
        \\ &= \lambda^{-1} \left[ B_{st} + (0 - B_{st}) \right] = 0.
    \end{align*}

    ($\Longrightarrow$):
    Suppose $G \homb H$.
    By \cref{thm:gram}, there is a matrix $C \succeq 0$ satisfying
    properties~\eqref{eq:Cb_first}-\eqref{eq:Cb_last}.
    Let $Z$ achieve the optimal value (call it $\lambda$)
    for the minimization~\eqref{eq:th_min_sdp} for $\thbar(H)$.
    We will provide a feasible solution for~\eqref{eq:th_min_sdp} for $\thbar(G)$ to show that
    $\thbar(G) \le \lambda = \thbar(H)$.
    To this end, let $\ket{\mathbf{1}}$ be the all ones vector and define
    \begin{align*}
    Y = \big(I \otimes \bra{\mathbf{1}}\big) \big[\big(J \otimes Z\big) \circ C\big] \big(I \otimes \ket{\mathbf{1}} \big).
    \end{align*}
    Since $C \succeq 0$ and $Z \succeq 0$, and positive semidefiniteness is preserved by conjugation, we have that $Y \succeq 0$. Also note that
    \[Y_{xy} = \sum_{st} Z_{st} C_{xyst}.\]
    Using the fact that $Z_{ss} = \lambda-1$ and $C_{xxst}=0$ for $s \ne t$, we have
    \begin{align*}
        Y_{xx} &= \sum_{st} Z_{st} C_{xxst}
        = (\lambda-1) \sum_{st} C_{xxst} = \lambda-1.
    \end{align*}
    Using the fact that $Z_{st} = -1$ for $s \simH t$ and
    $C_{xyst} = 0$ for $x \simG y$, $s \not\simH t$, we have
    that for $x \simG y$,
    \begin{align*}
        Y_{xy} &= \sum_{st} Z_{st} C_{xyst}
        = \sum_{s \simH t} Z_{st} C_{xyst}
        = (-1)\sum_{s \simH t} C_{xyst}
        \\ &= (-1)\sum_{st} C_{xyst}
        = -1.
    \end{align*}
    Now define a matrix $Y'$ consisting of the real part of $Y$ (i.e.\ with coefficients
    $Y_{xy}' = \textrm{Re}[Y_{xy}]$).
    This matrix is real, positive semidefinite,\footnote{
        The entrywise complex conjugate of a positive semidefinite matrix is positive
        semidefinite, so $Y' = (Y+\textrm{conj}(Y))/2 \succeq 0$.
    } and satisfies $Y_{xx}=\lambda-1$ for
    all $x$ and $Y_{xy} = -1$ for $x \sim y$.
    Therefore $Y'$ is feasible for~\eqref{eq:th_min_sdp} with value $\lambda=\thbar(H)$.
    Since $\thbar(G)$ is the minimum possible value of~\eqref{eq:th_min_sdp}, we have
    $\thbar(G) \le \thbar(H)$.
\end{proof}

We are now prepared to answer in the affirmative an open question posed by
Beigi~\cite{PhysRevA.82.010303}.

\begin{corollary}
    \label{thm:beigi_beta}
    Let $\beta(\Hc)$ be the largest $n$ such that there exist vectors
    $\ket{w} \ne 0$ and $\ket{w_s^x}$ with $x \in \{1,\dots,n\}$ and $s \in V(H)$
    which satisfy conditions~\eqref{eq:beigi_first}-\eqref{eq:beigi_last}.
    Then $\beta(\Hc) = \lfloor \thbar(H) \rfloor$.
\end{corollary}
\begin{proof}
    Considering $K_n \homb H$, the conditions of \cref{def:hombpv} are equivalent
    to~\eqref{eq:beigi_first}-\eqref{eq:beigi_last}.
    Since $\thbar(K_n) = n$,
    \cref{thm:th_homb} gives $K_n \homb H \iff n \le \thbar(H)$.
    Since $\beta(\Hc)$ is the largest $n$ such that $K_n \homb H$, we have that
    $\beta(\Hc) = \lfloor \thbar(H) \rfloor$.
\end{proof}

A related corollary can be formed by considering $G \homb K_n$ rather than $K_n \homb H$.
This defines a set of vectors $\ket{w_s^x}$ satisfying conditions in some sense complementary
to Beigi's~\eqref{eq:beigi_first}-\eqref{eq:beigi_last}.
Now we approach $\vartheta$ from above:

\begin{corollary}
    \label{thm:beigi_beta_chi}
    Let $\beta_\chi(G)$ be the smallest $n$ such that there exist vectors
    $\ket{w} \ne 0$ and $\ket{w_s^x}$ with $x \in V(G)$ and $s \in \{1,\dots,n\}$ for which
    \begin{enumerate}
        \item $\sum_s \ket{w_s^x} = \ket{w}$
        \item $\braket{w_s^x}{w_s^y} = 0$ for all $x \simG y$
        \item $\braket{w_s^x}{w_t^x} = 0$ for all $s \ne t$.
    \end{enumerate}
    Then $\beta_\chi(G) = \lceil \thbar(G) \rceil$.
\end{corollary}
\begin{proof}
    Considering $G \homb K_n$, the conditions of \cref{def:hombpv} are equivalent to the
    conditions stated above.
    Since $\thbar(K_n) = n$, \cref{thm:th_homb} gives
    $G \homb K_n \iff \thbar(G) \le n$.
    Since $\beta_\chi(G)$ is the smallest $n$ such that $G \homb K_n$, we have
    $\beta_\chi(G) = \lceil \thbar(G) \rceil$.
\end{proof}

\begin{corollary}
    \label{thm:ent_cost_rate_bound}
    The entanglement assisted cost rate is bounded as follows:
    \begin{align*}
        \eta^*(G, \Hc) \ge \frac{\log \thbar(G)}{\log \thbar(H)}.
    \end{align*}
\end{corollary}
\begin{proof}
    Since $\thbar(G^{\boxtimes m}) = \thbar(G)^m$~\cite{knuth94} and
    $\thbar(H^{\djp n}) = \thbar(H)^n$~\cite{lovasz79}, it follows that
    \begin{align*}
        G^{\boxtimes m} \home H^{\djp n}
            &\implies G^{\boxtimes m} \homb H^{\djp n}
            && \mbox{(by \cref{thm:hom_relax})}
        \\ &\implies \thbar(G^{\boxtimes m}) \le \thbar(H^{\djp n})
            && \mbox{(by \cref{thm:th_homb})}
        \\ &\implies \thbar(G)^m \le \thbar(H)^n
        \\ &\implies \frac{\log \thbar(G)}{\log \thbar(H)} \le \frac{n}{m}.
    \end{align*}
    Therefore,
    \begin{align*}
        \eta^*(G, \Hc) = \lim_{m \to \infty} \min_n \left\{
            \frac{n}{m} : G^{\boxtimes m} \home H^{\djp n}
            \right\} \ge \frac{\log \thbar(G)}{\log \thbar(H)}.
    \end{align*}
\end{proof}

Something similar to \cref{thm:th_homb} holds for the relation $G \homp H$.  In this case
there is an inequality not just for the Lov{\'a}sz $\vartheta$ number but also for
Schrijver's $\thm$ and Szegedy's $\thp$.
Unfortunately, this will no longer be an if-and-only-if statement (but see
\cref{thm:gives_homv} for a weakened converse, and \cref{sec:schrijver_iff} for a somewhat more
complicated if-and-only-if involving $\thmbar$).

\begin{theorem}
    \label{thm:th_homp}
    Suppose $G \homp H$.
    Then $\thbar(G) \le \thbar(H)$, $\thmbar(G) \le \thmbar(H)$, and $\thpbar(G) \le \thpbar(H)$.
\end{theorem}
\begin{proof}
    As per \cref{thm:gram}, let $C$ be a positive semidefinite matrix satisfying
    properties~\eqref{eq:Cb_first},~\eqref{eq:Cb_second}, and~\eqref{eq:Cp}.
    We give the proof for $\thmbar(G) \le \thmbar(H)$; the others are proved in a similar way.
    The proof is very similar to that of \cref{thm:th_homb}, with slight modification due to
    the fact that the last condition on $C$ is different.
    Let $Z$ achieve the optimal value for the minimization
    program~\eqref{eq:thm_min_sdp} for $\thmbar(H)$.
    We will provide a feasible solution for~\eqref{eq:thm_min_sdp} for $\thmbar(G)$ to show that
    $\thmbar(G) \le \thmbar(H)$.
    Specifically, let $Y_{xy} = \sum_{st} Z_{st} C_{xyst}$.
    Since $C$ and $Z$ are positive semidefinite, so is $Y$.

    A feasible solution for~\eqref{eq:thm_min_sdp}, with value $\thmbar(H)$,
    requires $Y_{xx} = \thmbar(H)-1$.  However, it suffices to show
    $Y_{xx} \le \thmbar(H)-1$ since equality can be achieved by adding a non-negative diagonal
    matrix to $Y$.  We have
    \begin{align*}
        Y_{xx} &= \sum_{st} Z_{st} C_{xxst}
        \\ &\le \max_{st} \abs{Z_{st}} \sum_{st} C_{xxst}
            & \mbox{(since $C_{xxst} \ge 0$)}
        \\ &\le \max_s \abs{Z_{ss}} \sum_{st} C_{xxst}
            & \mbox{(since $Z \succeq 0$)}
        \\ &= \thmbar(H)-1.
    \end{align*}
    Similarly, for $x \simG y$ we have
    \begin{align}
        \notag
        Y_{xy} &= \sum_{st} Z_{st} C_{xyst}
        = \sum_{s \simH t} Z_{st} C_{xyst}
        \\ &\le (-1) \sum_{s \simH t} C_{xyst}
        = (-1) \sum_{st} C_{xyst} = -1.
        \label{eq:Yxy_thm}
    \end{align}
    Therefore $Y$ is feasible for~\eqref{eq:thm_min_sdp} with value $\lambda=\thmbar(H)$.
    Since $\thmbar(G)$ is the minimum possible value of~\eqref{eq:thm_min_sdp}, we have
    $\thmbar(G) \le \thmbar(H)$.

    To show $\thbar(G) \le \thbar(H)$ or $\thpbar(G) \le \thpbar(H)$, replace inequality with equality
    in~\eqref{eq:Yxy_thm}.
    For $\thpbar(G) \le \thpbar(H)$ we have $Z_{st} \ge -1$ for all $s,t$ and need to show
    $Y_{xy} \ge -1$ for all $x,y$.  This is readily verified:
    \begin{align*}
        Y_{xy} &= \sum_{st} Z_{st} C_{xyst}
        \ge (-1) \sum_{st} C_{xyst}
        = -1.
    \end{align*}
\end{proof}

It is well known that
$\alpha(G) \le \thm(G) \le \vartheta(G) \le \thp(G) \le \chi(\Gc)$.
We show that similar inequalities hold for the entanglement assisted independence and chromatic
numbers.
Since $\thm$ and $\thp$ are not multiplicative under the required graph products
(\cref{sec:mult}), these do not lead to bounds on asymptotic quantities such as entanglement
assisted Shannon capacity or entanglement assisted cost rate.

\begin{corollary}
    $\alpha^*(H) \le \thm(H)$.
\end{corollary}
\begin{proof}
    By definition $\alpha^*(H)$ is the largest $n$ such that $K_n \home \Hc$.
    But $K_n \home \Hc \implies K_n \homp \Hc \implies \thmbar(K_n) \le \thm(H)$.
    Since $\thmbar(K_n) = n$, the conclusion follows.
\end{proof}
 The following corollary was already shown in~\cite{arxiv:1308.4283} via a different method.
\begin{corollary}
    Define $\chi^*(G)$ to be the smallest $n$ such that $G \home K_n$.
    Then $\chi^*(G) \ge \thpbar(G)$.
    \end{corollary}
\begin{proof}
    By definition, $\chi^*(G)$ is the smallest $n$ such that $G \home K_n$.
    But $G \home K_n \implies G \homp K_n \implies \thpbar(G) \le \thpbar(K_n) = n$.
\end{proof}

It would be nice to have a converse to \cref{thm:th_homp}, like there was with \cref{thm:th_homb}.
Is it the case that $\thbar(G) \le \thbar(H)$, $\thmbar(G) \le \thmbar(H)$,
and $\thpbar(G) \le \thpbar(H)$ together imply $G \homp H$?
We do not know.
However, it is the case that $\thpbar(G) \le \thmbar(H) \implies G \homp H$.
In fact, something stronger can be said.
We have the following theorem, the consequences of which will be further explored in
\cref{sec:qmhom}.

\begin{theorem}
    \label{thm:gives_homv}
    $\thpbar(G) \le \thmbar(H) \implies G \homv H$.
\end{theorem}
\begin{proof}
    The proof mirrors that of the ($\Longleftarrow$) portion of \cref{thm:th_homb}, so we only
    describe the differences.
    Let $\lambda = \thmbar(H)$.
    \Cref{thm:thm_max_IT} in \cref{sec:schrijver_iff} gives that
    \begin{align*}
        \thmbar(H) &= \max\{ \opnorm{I+T} : I + T \succeq 0,
            \notag \\ &\hphantom{= \max\{\;}
            T_{st} = 0 \textnormal{ for } s \not\sim t,
            \notag \\ &\hphantom{= \max\{\;}
            T_{st} \ge 0 \textnormal{ for all } s,t \}.
    \end{align*}
    So there is a matrix $T$ such that $\opnorm{I+T}=\lambda$,
    $I+T \succeq 0$, $T_{st}=0$ for $s \not\sim t$, and $T_{st} \ge 0$ for all $s,t$.
    Since $\lambda \ge \thpbar(G)$, there is a matrix $Z$ such that
    $Z \succeq 0$, $Z_{xx} = \lambda-1$ for all $x$,
    $Z_{xy} = -1$ for all $x \simG y$, and $Z_{xy} \ge -1$ for all $x,y$.
    Note that $T$ and $Z$ satisfy all conditions required by \cref{thm:th_homb} plus the
    additional conditions $T_{st} \ge 0$ for all $s,t$ and $Z_{xy} \ge -1$ for all $x,y$.

    Define $B$ and $D$ as in \cref{thm:th_homb}.
    The eigenvector $\ket{\psi}$ corresponding to the maximum eigenvalue of $I+T$ can be
    chosen to be entrywise non-negative (this follows from the Perron--Frobenius theorem
    and the fact that $I+T$ is entrywise non-negative).
    It follows that $B$ can be chosen entrywise non-negative.
    As before, define
    \begin{align*}
        C = \lambda^{-1} \left[ J \ot B + (\lambda-1)^{-1} Z \ot (\lambda D - B) \right].
    \end{align*}
    Since $T$ and $Z$ satisfy all conditions needed by \cref{thm:th_homb}, $C$
    satisfies~\eqref{eq:Cb_first}-\eqref{eq:Cb_last}.
    To get $G \homv H$ it remains only to show satisfaction of~\eqref{eq:Cp}:
    $C_{xyst} \ge 0$ for all $x,y,s,t$.
    When $s=t$,
    \begin{align*}
        C_{xyss} &= \lambda^{-1} [ D_{ss} + (\lambda-1)^{-1} Z_{xy} (\lambda D_{ss} - D_{ss}) ]
        \\ &= \lambda^{-1} (1 + Z_{xy}) D_{ss} \ge 0.
    \end{align*}
    The last inequality follows from $Z_{xy} \ge -1$ and $D_{ss} \ge 0$.
    When $s \ne t$,
    \begin{align*}
        C_{xyst} &= \lambda^{-1} [ B_{st} + (\lambda-1)^{-1} Z_{xy} (0 - B_{st}) ]
        \\ &= \lambda^{-1} (\lambda-1)^{-1} [(\lambda-1) - Z_{xy}] B_{st} \ge 0.
    \end{align*}
    The last inequality follows from $Z_{xy} \le \max\{ Z_{xx}, Z_{yy} \} = \lambda-1$
    (since $Z \succeq 0$) and $B_{st} \ge 0$.
\end{proof}

Finally, we show that the two conditions $G \homp H$ and $G \homb H$ are not equivalent: the
second one is weaker.

\begin{theorem}
    \label{thm:homp_vs_home}
    If $G \homp H$ then $G \homb H$, but there are graphs for which the converse does not hold.
\end{theorem}
\begin{proof}
    The forward implication is an immediate consequence of \cref{thm:th_homb,thm:th_homp}:
    \begin{align*}
        G \homp H \implies \thbar(G) \le \thbar(H) \implies G \homb H.
    \end{align*}
    To see that the converse does not hold, take $H$ to be any graph such that
    $\lfloor \thmbar(H) \rfloor < \lfloor \thbar(H) \rfloor$.
    For example, a graph with
    $\thmbar(H) = 4$ but $\thbar(H) = 16/3 > 5$ is given at the end of~\cite{1056072}.
    Then $5 = \thbar(K_5) \le \thbar(H) \implies K_5 \homb H$
    but $5 = \thmbar(K_5) > \thmbar(H) \implies K_5 \not\homp H$.
\end{proof}

\section{Quantum homomorphisms}
\label{sec:qmhom}

Suppose Alice and Bob share an entangled state $\ket{\psi} \in \hilb{A} \ot \hilb{B}$ on
Hilbert spaces of arbitrary dimension.
A referee asks Alice a question $x \in X$ and Bob a question $y \in Y$.
Based on $x$, Alice performs a (without loss of generality, projective)
measurement $\{E_s^x\}_s$ and reports outcome $s \in S$ to the referee.
Similarly, Bob performs measurement $\{F_t^y\}_t$ and reports $t \in T$.
The sets $X,Y,S,T$ are finite.
The probability distribution of Alice and Bob's answer, conditioned upon the referee's
question, is
\begin{align}
    \label{eq:behaviorQ}
    P(s,t|x,y) = \braopket{\psi}{E_s^x \ot F_t^y}{\psi}
\end{align}
where $\sum_s E_s^x = I$, $\sum_t F_t^y = I$, and $\braket{\psi}{\psi}=1$.

The assumption that Alice and Bob's measurements take such a tensor product form is
associated with non-relativistic quantum mechanics.  One may alternatively consider a
model in which there is only a single Hilbert space, $\ket{\psi} \in \hilb{A}$ and
$E_s^x, F_t^y \in \linop{\hilb{A}}$, but in which each $E_s^x$ commutes with each $F_t^y$.
The conditional probability distribution in this model is
\begin{align}
    \label{eq:behaviorQc}
    P(s,t|x,y) = \braopket{\psi}{E_s^x F_t^y}{\psi}.
\end{align}
Tsirelson's problem is the question of whether these two models differ.  That is to say,
is there a conditional probability distribution realizable as~\eqref{eq:behaviorQc} but
not as~\eqref{eq:behaviorQ}?
Tsirelson showed that if the Hilbert spaces are finite dimensional then the two models
are the same.  A simplified proof appears in~\cite{arxiv:0812.4305}.
In addition to its importance to quantum mechanics, Tsirelson's problem is of mathematical
interest since it is closely related to Connes' embedding problem~\cite{junge2011}.
Note that any correlation of the form~\eqref{eq:behaviorQ} can be written in the
form~\eqref{eq:behaviorQc} since $\braopket{\psi}{E_s^x \ot F_t^y}{\psi} =
\braopket{\psi}{(E_s^x \ot I)(I \ot F_t^y)}{\psi}$ and
$E_s^x \ot I$ commutes with $I \ot F_t^y$.

Graph $G$ is said to have a \emph{quantum homomorphism} to $H$ (written $G \homq H$)
if there is a probability distribution of the form~\eqref{eq:behaviorQ}, with
$X=Y=V(G)$, $S=T=V(H)$, and finite dimensional $\ket{\psi}$,
satisfying~\cite{arxiv:1212.1724}
\begin{align}
    P(s \ne t|x=y) &= 0 \notag \\
    P(s \not\sim_H t|x \sim_G y) &= 0.
    \label{eq:homgameP}
\end{align}
This is called a ``quantum homomorphism'' because if Alice and Bob are not allowed to
share an entangled state [equivalently, if $\dim(\hilb{A}) = \dim(\hilb{B}) = 1$] then
such a conditional probability distribution is achievable if and only if $G \to H$.
Although $G \homq H \implies G \home H$, it is an open question whether the converse
holds.

In~\cite{arxiv:1212.1724} it is shown that $G \homq H$ if and only if
there exist projection operators (i.e., Hermitian matrices with eigenvalues in $\{0,1\}$)
$E_s^x$ for $x \in V(G)$ and $s \in V(H)$ such that
\begin{align*}
    \sum_s E_s^x &= I \\
    E_s^x E_t^y &= 0 \textrm{ for all } x \simG y, s \not\simH t \\
    E_s^x E_t^x &= 0 \textrm{ for all } s \ne t.
\end{align*}
Note that the first condition actually implies the third.
Define $\ket{w_s^x} = \vecize(E_s^x)$.
Since $\braket{w_s^x}{w_t^y} = \Tr(E_s^x E_t^y)$, this gives a set of vectors
satisfying the conditions of \cref{def:hombpv} for $G \homv H$.
So $G \homq H \implies G \homv H$.
Since $G \homv H \implies G \homp H$,
\cref{thm:th_homp} gives the following corollary which was previously shown
in~\cite{roberson2013variations}.
\begin{corollary}
    Suppose $G \homq H$.
    Then $\thbar(G) \le \thbar(H)$, $\thmbar(G) \le \thmbar(H)$, and $\thpbar(G) \le \thpbar(H)$.
\end{corollary}

The \textit{quantum chromatic number} $\chi_q(G)$ is the least $n$ such that
$G \homq K_n$, in analogy to the chromatic number $\chi(G)$ which is the least $n$ such
that $G \homm K_n$.
As a means of studying Tsirelson's problem, a number of variations of $\chi_q$ were
considered in~\cite{arxiv:1311.6850}.
For instance, they define $\chiqr$ by taking the correlation model to
be~\eqref{eq:behaviorQc} with infinite dimensional $\ket{\psi}$ rather
than~\eqref{eq:behaviorQ} with finite dimensional $\ket{\psi}$ as was used to define
$G \homq H$ (and thus $\chi_q$).
Also, they consider a semidefinite relaxation $\chivect$.
In the language of our paper, $\chivect(G)$ is the least $n$ such that
$G \homv K_n$.  In fact, our $G \homv H$ definition was inspired by their work\footnote{
    The first version of the present paper was posted before~\cite{arxiv:1311.6850}.  We
    later amended this paper to address the question posed in~\cite{arxiv:1311.6850}.
}.
One could also define $\omegavect(H)$ as the largest $n$ such that
$K_n \homv H$.
Both $\chivect(G)$ and $\omegavect(H)$ can be computed using the tools of
\cref{sec:mainthms}.

\begin{corollary}
    \label{thm:chivect_equals}
    $\chivect(G) = \lceil \thpbar(G) \rceil$ and
    $\omegavect(H) = \lfloor \thmbar(H) \rfloor$.
\end{corollary}
\begin{proof}
    \Cref{thm:gives_homv} gives (for integer $n$) $\thpbar(G) \le n \implies G \homv K_n$
    and $n \le \thmbar(H) \implies K_n \homv H$.
    \Cref{thm:th_homp} gives the converse, so $\thpbar(G) \le n \iff G \homv K_n$
    and $n \le \thmbar(H) \iff K_n \homv H$.
\end{proof}

The authors of~\cite{arxiv:1311.6850} posed the question of whether
$\chivect(G) = \chi_q(G)$, that is to say whether $\chi_q$ is equivalent to its
semidefinite relaxation.  In fact, these two quantities are not equal.

\begin{theorem}
    \label{thm:homq_homv_gap}
    There is a graph $G$ such that $\chivect(G) < \chi_q(G)$.
    Therefore $G \homv H$ does not imply $G \homq H$.
\end{theorem}
\begin{proof}
    In light of \cref{thm:chivect_equals}, the goal is to find $G$ such that
    $\lceil \thpbar(G) \rceil < \chi_q(G)$.
    The projective rank of a graph, $\xi_f(G)$, is the infimum of $d/r$ such that the vertices of
    a graph can be assigned rank-$r$ projectors in $\mathbb{C}^d$ such that adjacent
    vertices have orthogonal projectors.
    Since $\xi_f(G) \le \chi_q(G)$~\cite{roberson2013variations}, it suffices
    to find a gap between $\lceil \thpbar(G) \rceil$ and $\xi_f(G)$.

    The five cycle has $\thpbar(C_5)=\sqrt{5} < \xi_f(C_5) = 5/2$~\cite{roberson2013variations}.
    But this is not enough since $\lceil \sqrt{5} \rceil = 3 > 5/2$.
    Fortunately, we can amplify the difference by taking the disjunctive product with a
    complete graph.
    $\thpbar$ is sub-multiplicative under disjunctive product, as feasible solutions $Z+J$
    to~\eqref{eq:thp_min_sdp} can be combined by tensor product.
    \Cref{thm:xif_mult} states that $\xi_f$ is multiplicative under the disjunctive (and
    lexicographical) product, so
    \begin{align*}
        \left\lceil \thpbar(C_5 \djp K_3) \right\rceil
        \le \left\lceil 3\sqrt{5} \right\rceil = 7 < 3 \cdot \frac{5}{2}
        = \xi_f(C_5 \djp K_3).
    \end{align*}

\end{proof}

Subsequently, this result has been strengthened to
$\chivect(G) < \chiqr(G)$~\cite{arxiv:1407.6918}.

\section{Conclusion}

Beigi provided a vector relaxation of the entanglement assisted zero-error communication
problem, leading to an upper bound on the entanglement assisted independence number:
$\alpha^* \le \lfloor \vartheta \rfloor$~\cite{PhysRevA.82.010303}.
We generalized Beigi's construction to apply it to entanglement assisted zero-error
source-channel coding, defining a relaxed graph homomorphism $G \homb H$.
This ends up exactly characterizing monotonicity of $\vartheta$, and shows that $\vartheta$ can be
used to provide a lower bound on the cost rate for entanglement assisted source-channel coding.
As a corollary we answer in the affirmative an open question posed by Beigi of whether a
quantity $\beta$ that he defined is equal to $\lfloor \vartheta \rfloor$.
Applying a Beigi-style argument to chromatic number
rather than independence number yields a quantity analogous to $\beta$ which is equal to
$\lceil \vartheta \rceil$.
We defined a similar (and stronger) relaxation, $G \homp H$, which yields bounds involving
Schrijver's number $\thm$ and Szegedy's number $\thp$.
This leads to a stronger bound on entanglement assisted independence number:
$\alpha^* \le \lfloor \thm \rfloor$.
In addition to these new bounds, we reproduce previously known bounds
from~\cite{PhysRevA.82.010303,arxiv:1308.4283,arxiv:1212.1724,roberson2013variations}.
We also answer an open question from~\cite{arxiv:1311.6850} regarding the relation of the
quantum chromatic number to its semidefinite relaxation.

A number of open questions remain.
Since there is a graph for which $\thm < \vartheta-1$~\cite{1056072}, our bound
$\alpha^* \le \lfloor \thm \rfloor$ shows a gap between one-shot entanglement assisted
zero-error capacity and $\lfloor \vartheta \rfloor$.
However, since $\thm$ is not multiplicative, it is still not known whether there can be a gap
between the \emph{asymptotic} capacity (i.e.\ the entanglement assisted Shannon capacity)
and $\vartheta$.
To show such a gap requires a stronger bound on entanglement assisted Shannon capacity.
Haemers provided a bound on Shannon capacity which is sometimes stronger than
Lov{\'a}sz' bound~\cite{haemers1978,haemers1979,alon1998,peeters1996}; however, this bound
does not apply to the entanglement assisted case~\cite{leung2012entanglement}.

%

\begin{figure}[ht]
    \centering
    \includegraphics[scale=1.0]{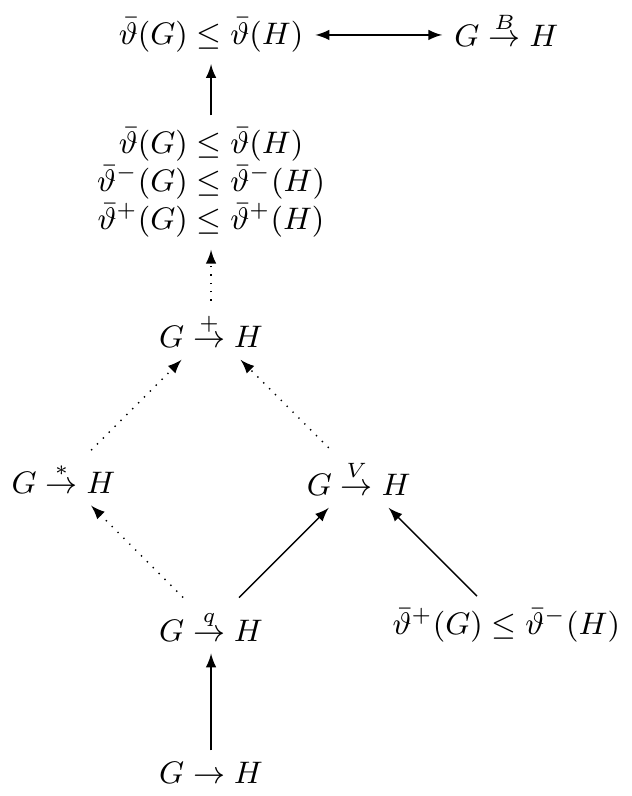}
    \caption{
        Implications between various conditions discussed in this paper.  Double ended arrows mean
        if-and-only-if, solid arrows mean the converse is known to not hold, and dotted
        arrows mean we do not know whether the converse holds.
    }
    \label{fig:implications}
\end{figure}

The standard notion of graph homomorphism, along with two of its quantum generalizations, and
our three relaxations, form a hierarchy as outlined in \cref{fig:implications}.
In some cases we do not know whether converses hold.
$G \home H$ is equivalent to $G \homq H$ if and only if projective measurements and a maximally
entangled state always suffice for entanglement assisted
zero-error source-channel coding.
Equivalence between $G \home H$ and $G \homp H$ seems unlikely but would have two important
consequences.
First, we would have a much simpler characterization (vector rather than operator) of
entanglement assisted homomorphisms and, in particular, entanglement assisted zero-error
communication.
Second, since $G \homv H \implies G \homp H$, the gap that we found between $G \homq H$ and
$G \homv H$ would give a gap between $G \homq H$ and $G \home H$.

After completing this work, we became aware of a previous investigation of a similar
problem.  Semidefinite relaxations of the homomorphism game (outside of the quantum
context) were investigated
in~\cite{Feige:1992:TOP:129712.129783}, and this was further developed
in~\cite{bacikmahajan}.
What they call a \emph{hoax} corresponds to our $G \homv H$, and what they call a
\emph{semi-hoax} corresponds to our $G \homb H$.
Though they studied the same problem, they reached a different conclusion:
they showed (using our terminology)
\begin{align}
    G \homb H &\iff \thbar(G \circ H) = \abs{V(G)}
    \label{eq:bacikB}
    \\
    G \homv H &\iff \thmbar(G \circ H) = \abs{V(G)}
    \label{eq:bacikV}
\end{align}
where $\circ$ denotes the \emph{hom-product} with vertices $V(G) \times V(H)$ and edges
\begin{align*}
    (x,s) \sim (y,t) \iff
    &(x \ne y) \textrm{ and }
    (x \sim y \implies s \sim t).
\end{align*}
Combining our \cref{thm:th_homb} with~\eqref{eq:bacikB} gives
$\thbar(G \circ H) = \abs{V(G)} \iff \thbar(G) \le \thbar(H)$ and
combining \cref{thm:th_homp,thm:gives_homv} with~\eqref{eq:bacikV} gives
\begin{align*}
    \thpbar(G) \le \thmbar(H) &\implies \thmbar(G \circ H) = \abs{V(G)}
    \\ &\implies \thbar(G) \le \thbar(H), \thmbar(G) \le \thmbar(H),
    \\ &\hphantom{{}\implies{}} \thpbar(G) \le \thpbar(H).
\end{align*}
Considering the special case $H=K_n$, we have $G \circ K_n=G \cart K_n$ (Cartesian
product) giving $\thbar(G \cart K_n) = \abs{V(G)} \iff \thbar(G) \le n$ and
$\thmbar(G \cart K_n) = \abs{V(G)} \iff \thpbar(G) \le n$, reproducing Theorem 2.7 of
\cite{doi:10.1137/050648237}.

\section{Acknowledgments}
We thank Vern Paulsen and Ivan Todorov for inspiring us to define $G \homv H$ and to find the
gap between this and $G \homq H$.

TC is supported by the Royal Society.
LM is supported by the Ministry of Education (MOE) and National Research Foundation Singapore, 
as well as MOE Tier 3 Grant ``Random numbers from quantum processes'' (MOE2012-T3-1-009).
DR is supported by an NTU start-up grant awarded to D.V.~Pasechnik.
SS is supported by the Royal Society and the British Heart Foundation.
DS is supported by the National Science Foundation through Grant PHY-1068331.
AW is supported by the European Commission
(STREPs ``QCS'' and ``RAQUEL''), the European Research Council (Advanced Grant
``IRQUAT'') and the Philip Leverhulme Trust; furthermore by the Spanish MINECO,
project FIS2008-01236, with the support of FEDER funds.




\appendices

\section{Multiplicativity}
\label[secinapp]{sec:mult}

In Lov\'{a}sz' original paper~\cite{lovasz79} on the $\vartheta$ function, he proved that
\[\vartheta(G \boxtimes H) = \vartheta(G) \vartheta(H),\]
i.e.~$\vartheta$ is multiplicative with respect to the strong product. To do this he proved the following two inequalities:
\[\vartheta(G)\vartheta(H) \le \vartheta(G \boxtimes H) \le \vartheta(G) \vartheta(H).\]
This sufficed for Lov\'{a}sz because it was only required to show that $\vartheta$ is multiplicative with respect to the strong product in order to prove that it was an upper bound on Shannon capacity. However, Lov\'{a}sz also noted that his proof of the first inequality above also proves the following stronger statement:
\[\vartheta(G)\vartheta(H) \le \vartheta(G * H).\]
Together these inequalities imply that $\vartheta$ is multiplicative with respect to both
the strong and disjunctive products. Our aim in this appendix is to show that
$\thm$ is not multiplicative with respect to the strong product and $\thp$ is not
multiplicative with respect to the disjunctive product, but $\thm$
is multiplicative with respect to the disjunctive product.
Also we will show multiplicativity of projective rank $\xi_f$.

\subsection{Counterexamples}

Some of the inequalities involving $\vartheta$ above can be proved for $\thm$ as well. 
Adapting Lov\'{a}sz' proof of the analogous statement for $\vartheta$, it can be shown that
\[ \thm(G \boxtimes H) \ge \thm(G \djp H) \ge \thm(G)\thm(H).\]
Similarly, it can be shown that 
\[ \thp(G \djp H) \le \thp(G \boxtimes H) \le \thp(G)\thp(H).\]
Therefore, in order to show that neither $\thm$ or $\thp$ are multiplicative with 
respect to both the strong and disjunctive products, we must find counterexamples 
to both of the following inequalities:
\[\thm(G \boxtimes H) \le \thm(G)\thm(H), \quad \thp(G \djp H) \ge \thp(G)\thp(H).\]

At the end of~\cite{1056072}, Schrijver gives an example of a graph, which we refer 
to as $G_S$, that satisfies $\thm(G_S) < \vartheta(G_S)$. The vertices of $G_S$ are the 
$0$-$1$-strings of length six, and two strings are adjacent if their Hamming distance 
is at most three. In other words, Schrijver's graph $G_S$ is an instance of a 
\emph{Hamming graph}. 
Note that this graph is vertex transitive. We will see how to use the graph $G_S$ 
to construct counterexamples to both of the above inequalities. To do this we will 
need two lemmas, the first of which is from~\cite{lovasz79}.

\begin{lemma}\label{lem:lovaszrecip}
For any graph $G$,
\[\vartheta(G)\vartheta(\overline{G}) \ge |V(G)|,\]
with equality when $G$ is vertex transitive.
\end{lemma}

An analogous statement involving $\thm$ and $\thp$ was proved by Szegedy in~\cite{365707}:

\begin{lemma}\label{lem:szegedyrecip}
For any graph $G$,
\[\thm(G)\thp(\overline{G}) \ge |V(G)|,\]
with equality when $G$ is vertex transitive.
\end{lemma}

One easy consequence of these lemmas is that if $G$ is a vertex transitive graph such that $\thm(G) < \vartheta(G)$, then
\[\thp(\overline{G}) = \frac{|V(G)|}{\thm(G)} > \frac{|V(G)|}{\vartheta(G)} = \vartheta(\overline{G}).\]
In particular this implies that $\thp(\overline{G_S}) > \vartheta(\overline{G_S})$.

More pertinent to our discussion are the following lemmas.

\begin{lemma}\label{lem:thmcounter}
If $G$ is a vertex transitive graph such that $\thm(G) < \vartheta(G)$, then
\[\thm(G \boxtimes \overline{G}) > \thm(G)\thm(\overline{G}).\]
\end{lemma}
\begin{proof}
First note the vertices of the form $(v,v)$ in $G \boxtimes \overline{G}$ form an independent set of size $|V(G)|$. Therefore, $\thm(G \boxtimes \overline{G}) \ge |V(G)|$, and we have the following:
\[\thm(G)\thm(\overline{G}) < \vartheta(G)\vartheta(\overline{G}) = |V(G)| \le \thm(G \boxtimes \overline{G}),\]
since $\thm(\overline{G}) \le \vartheta(\overline{G})$.
\end{proof}
Since $G_S$ satisfies the hypotheses of~\cref{lem:thmcounter}, we have the following desired corollary:
\begin{corollary}
The parameter $\thm$ is not multiplicative with respect to the strong product.
\end{corollary}

We are also able to use \cref{lem:thmcounter} to prove a similar lemma for $\thp$.
\begin{lemma}
If $G$ is a vertex transitive graph such that $\thm(G) < \vartheta(G)$, then
\[\thp(G \djp \overline{G}) < \thp(G)\thp(\overline{G}).\]
\end{lemma}
\begin{proof}
Suppose that $G$ is such a graph. By \cref{lem:thmcounter}, we have that
\[\thm(G \boxtimes \overline{G}) > \thm(G)\thm(\overline{G}).\]
Since $G$ is vertex transitive, so is $G \djp \overline{G}$ and thus we can apply
\cref{lem:szegedyrecip} to obtain
\begin{align*}
    \thp(G \djp \overline{G})
    &= \frac{|V(G)|^2}{\thm\left(\overline{G \djp \overline{G}}\right)}
    = \frac{|V(G)|^2}{\thm(\overline{G} \boxtimes G)}
    \\ &< \frac{|V(G)|}{\thm(\overline{G})} \frac{|V(G)|}{\thm(G)} = \thp(G) \thp(\overline{G}).
\end{align*}
\end{proof}
Similarly to the above, this implies the following:
\begin{corollary}
The parameter $\thp$ is not multiplicative with respect to the disjunctive product.
\end{corollary}

Even though $\thm$ is not multiplicative with respect to the strong product, nor is $\thp$ with
respect to the disjunctive product, one could ask whether they are at least multiplicative with
respect to the corresponding graph powers, as this would be enough to prove an analogue of
\cref{thm:ent_cost_rate_bound}.  It turns out that they are not, as we now show.
Non-multiplicativity for $\thm$ was shown already in~\cite{bomze2010gap} but with a much
smaller gap.

\begin{corollary}
    The parameter $\thm$ is not multiplicative under strong graph powers $G^{\boxtimes n}$, and
    $\thp$ is not multiplicative under disjunctive graph powers $G^{\djp n}$.
\end{corollary}
\begin{proof}
    Let $G_S$ be a vertex transitive graph such that $\thm(G_S) < \vartheta(G_S)$, whose
    existence was discussed above.
    Let $H = G_S \oplus \overline{G_S}$ where $\oplus$ denotes disjoint union.
    Since $\thm$ is additive under disjoint union and is super-multiplicative under the strong
    product,
    \begin{align*}
       \thm(H^{\boxtimes 2})
        &= \thm[ G_S^{\boxtimes 2} \oplus (G_S \boxtimes \overline{G_S})
            \oplus (\overline{G_S} \boxtimes G_S) \oplus \overline{G_S}^{\boxtimes 2} ]
        \\ &= \thm(G_S^{\boxtimes 2}) + \thm(G_S \boxtimes \overline{G_S})
            \\ &\phantom{=} + \thm(\overline{G_S} \boxtimes G_S) + \thm(\overline{G_S}^{\boxtimes 2})
        \\ &\ge \thm(G_S)^2 + \thm(G_S \boxtimes \overline{G_S})
            \\ &\phantom{=} + \thm(\overline{G_S} \boxtimes G_S) + \thm(\overline{G_S})^2
        \\ &> \thm(G_S)^2 + \thm(G_S) \thm(\overline{G_S})
            \\ &\phantom{=} + \thm(\overline{G_S}) \thm(G_S) + \thm(\overline{G_S})^2
        \\ &= [ \thm(G_S) + \thm(\overline{G_S} ]^2
        \\ &= \thm(H)^2.
    \end{align*}
    Similarly,
    \begin{align}
       \thp(H^{\djp 2})
        &= \thp[ (G_S \oplus \overline{G_S}) \djp (G_S \oplus \overline{G_S}) ] \notag
        \\ &\le \thp[ G_S^{\boxtimes 2} \oplus (G_S \djp \overline{G_S}) \notag
            \\ &\phantom{=} \oplus (\overline{G_S} \djp G_S) \oplus \overline{G_S}^{\boxtimes 2} ]
            \label{eq:used_thp_subg_mon}
        \\ &\le \thp(G_S)^2 + \thp(G_S \djp \overline{G_S}) \notag
            \\ &\phantom{=} + \thp(\overline{G_S} \djp G_S) + \thp(\overline{G_S})^2 \notag
        \\ &< \thp(G_S)^2 + \thp(G_S) \thp(\overline{G_S}) \notag
            \\ &\phantom{=} + \thp(\overline{G_S}) \thp(G_S) + \thp(\overline{G_S})^2 \notag
        \\ &= [ \thp(G_S) + \thp(\overline{G_S} ]^2 \notag
        \\ &= \thp(H)^2. \notag
    \end{align}
    where~\eqref{eq:used_thp_subg_mon} follows from the fact that $\thp(G_1) \ge \thp(G_2)$ when
    $G_1$ is a subgraph of $G_2$.
\end{proof}

\subsection{\texorpdfstring{$\thm$}{Schrijver's number} and the disjunctive product}

Though $\thm$ is not multiplicative with respect to the strong product, we are able to
show that it is multiplicative with respect to the disjunctive product and the
lexicographical product.
The lexicographical product $G[H]$ has vertices $V(G) \times V(H)$ and edges
$(x,y) \sim (x',y')$ if $x \sim_G x'$ or $(x=x' \textrm{ and } y \sim_H y')$.

\begin{theorem}
    Schrijver's number is multiplicative under the disjunctive and the lexicographical products:
    $\thm(G \djp H) = \thm(G[H]) = \thm(G) \thm(H)$.
    \label{thm:thmmulti}
\end{theorem}
\begin{proof}
    We use the following formulation for Schrijver's number:
    \begin{align}
        \thm(G) &= \max\{ \ip{B,J} : B \succeq 0, \Tr B = 1,
        \notag \\ &\hphantom{= \max\{\;} B_{ij} \ge 0 \textrm{ for all } i,j,
        \notag \\ &\hphantom{= \max\{\;} B_{ij}=0 \textrm{ for } i \sim j \}.
        \label{eq:thm_max_sdp}
    \end{align}
    It is easy to show that $\thm(G \djp H) \ge \thm(G)\thm(H)$: if $B_G$ and $B_H$ are
    optimal solutions of~\eqref{eq:thm_max_sdp} for $\thm(G)$ and
    $\thm(H)$ then $B_G \ot B_H$ is feasible for~\eqref{eq:thm_max_sdp} for
    $\thm(G \djp H)$ with value $\thm(G)\thm(H)$.
    Since $G[H]$ is a subgraph of $G \djp H$, we have $\thm(G[H]) \ge \thm(G \djp H)$.
    It remains only to show $\thm(G[H]) \le \thm(G)\thm(H)$.

    Let $B$ be an optimal solution for~\eqref{eq:thm_max_sdp} for $\thm(G[H])$.
    This can be considered as an operator $B \in \opV{G} \ot \opV{H}$,
    and we have $\ip{B, J_G \ot J_H} = \thm(G[H])$ where $J_G$ is the all ones matrix
    indexed by $V(G)$ and similarly for $J_H$.
    Note that $B_{xyx'y'}=0$ when $x \sim_G x'$ or $(x=x' \textrm{ and } y \sim_H y')$.
    For $x \in V(G)$ let $\ket{x}$ denote the corresponding basis vector
    in $\mathbb{C}^{\abs{V(G)}}$ and define
    \begin{align*}
        B^x = (\bra{x} \ot I) B (\ket{x} \ot I) \in \opV{H}.
    \end{align*}
    Since $B^x \succeq 0$, $B^x_{yy'} \ge 0$ for all $y,y'$, and
    $B^x_{yy'} = 0$ when $y \simH y'$, it holds that
    $B^x/\Tr B^x$ is feasible for~\eqref{eq:thm_max_sdp} for $\thm(H)$.
    The value of this solution is $\ip{B^x,J_H}/\Tr B^x$, thus
    \begin{align*}
        \frac{\ip{B^x,J_H}}{\Tr B^x} \le \thm(H).
    \end{align*}
    Also, $\sum_x \Tr B^x = \Tr B = 1$, giving
    \begin{align*}
        \sum_x \ip{B^x,J_H} \le \sum_x \thm(H) \Tr B^x = \thm(H).
    \end{align*}
    Define $B' = (I \ot \bra{\textbf{1}}) B (I \ot \ket{\textbf{1}}) \in \opV{G}$
    where $\ket{\textbf{1}}$ is the all ones vector.
    Note that $B' = \Tr_H\{ (I \ot J_H) B \}$.
    Since $B' \succeq 0$, $B'_{xx'} \ge 0$ for all $x,x'$, and
    $B'_{xx'} = 0$ when $x \simG x'$, it holds that
    $B'/\Tr B'$ is feasible for~\eqref{eq:thm_max_sdp} for $\thm(G)$.
    The value of this solution is $\ip{B',J_G}/\Tr B'$, thus
    $\ip{B',J_G}/\Tr B' \le \thm(G)$.
    Finally,
    \begin{align*}
        \thm(G[H]) &= \ip{B, J_G \ot J_H} = \ip{B', J_G}
        \\ &\le \thm(G) \Tr B'
        \\ &= \thm(G) \sum_x \ip{B^x, J_H}
        \\ &\le \thm(G) \thm(H).
    \end{align*}
\end{proof}

\subsection{What About \texorpdfstring{$\thp$}{Szegedy's number}?}

Based on other results concerning $\thm$ and $\thp$, \cref{thm:thmmulti} seems to
suggest that one should be able to prove that $\thp$ is multiplicative with respect to the
strong product. We already noted above that one of the needed inequalities, namely $\thp(G
\boxtimes H) \le \thp(G)\thp(H)$, does hold, so we would only need to show that $\thp(G
\boxtimes H) \ge \thp(G)\thp(H)$ holds as well. For now, a proof of this fact eludes us,
but we are able to prove the multiplicativity of $\thp$ in the case of vertex transitive
graphs using \cref{lem:szegedyrecip} and the multiplicativity of $\thm$ with respect to the disjunctive product.

\begin{theorem}
If $G$ and $H$ are vertex transitive, then 
\[\thp(G \boxtimes H) = \thp(G) \thp(H).\]
\end{theorem}
\begin{proof}
Since $G$ and $H$ are vertex transitive, so is $G \boxtimes H$. Therefore
\begin{align*}
    \thp(G \boxtimes H) &= \frac{|V(G)|\cdot |V(H)|}{\thm(\overline{G} \djp \overline{H})}
    = \frac{|V(G)|\cdot |V(H)|}{\thm(\overline{G}) \thm(\overline{H})}
    \\ &= \thp(G) \thp(H).
\end{align*}
\end{proof}

This seems to be pretty strong evidence that $\thp$ is multiplicative with respect to the strong product in general.

\subsection{Projective Rank}

The \emph{projective rank} of a graph, $\xi_f(G)$, is the infimum of $d/r$ such that the
vertices of a graph can be assigned rank-$r$ projectors in $\mathbb{C}^d$ such that
adjacent vertices have orthogonal projectors.  Such an assignment is called a
$d/r$-\textit{representation}. The `$f$' subscript in the notation for projective rank indicates that it can be thought of as a fractional version of orthogonal rank: the minimum dimension of an assignment of vectors such that adjacent vertices receive orthogonal vectors.
We will show $\xi_f$ to be multiplicative under both the disjunctive and the lexicographical
products.
As a reminder, the lexicographical product $G[H]$ has edges
$(x,y) \sim (x',y')$ if $x \sim_G x'$ or $(x=x' \textrm{ and } y \sim_H y')$.

\begin{theorem}
    Projective rank is multiplicative under the disjunctive and the lexicographical products:
    $\xi_f(G \djp H) = \xi_f(G[H]) = \xi_f(G) \xi_f(H)$.
    \label{thm:xif_mult}
\end{theorem}
\begin{proof}
    A $d_1/r_1$-representation for $G$ and a $d_2/r_2$-representation
    for $H$ can be turned into a $d_1 d_2/r_1 r_2$-representation for $G \djp H$ by taking the
    tensor products of the projectors associated with each graph.
    So $\xi_f(G \djp H) \le \xi_f(G) \xi_f(H)$.

    On the other hand, let $U_{xy}$ for $x \in V(G), y \in V(H)$ be the subspaces associated with
    a $d/r$-representation of $G[H]$.
    For each $x$, the subspaces $\{ U_{xy} : y \}$ form an $r'_x/r$ projective representation
    of $H$ where $r'_x$ is the dimension of $\linspan\{ U_{xy} : y \}$, so it must hold that
    $r'_x/r \ge \xi_f(H)$.
    Let $r' = \min\{r'_x\}$ and for each $x$ let $V_x$ be an $r'$ dimensional subspace
    of $\linspan\{ U_{xy} : y \}$.
    These form a $d/r'$ representation of $G$, so $d/r' \ge \xi_f(G)$.
    Then, $d/r = (d/r')(r'/r) \ge \xi_f(G) \xi_f(H)$
    so $\xi_f(G[H]) \ge \xi_f(G) \xi_f(H)$.
    Since $G[H] \subseteq G \djp H$ we have
    $\xi_f(G[H]) \ge \xi_f(G) \xi_f(H) \ge \xi_f(G \djp H) \ge \xi_f(G[H])$.
\end{proof}

\section{An if-and-only-if for Schrijver's number}
\label[secinapp]{sec:schrijver_iff}

Monotonicity of Schrijver's number admits an if-and-only-if statement along the lines of
\cref{thm:th_homb}; however, the corresponding conditions on the $\ket{w_s^x}$ vectors are a
bit more complicated and there is seemingly no direct connection to entanglement assisted
source-channel coding.
Specifically, we have the following result:
\begin{theorem}
    \label{thm:sch_mon_iff}
    $\thmbar(G) \le \thmbar(H)$ if and only if there are vectors $\ket{w} \ne 0$
    and $\ket{w_s^x} \in
    \mathbb{C}^d$ for each $x \in V(G)$, $s \in V(H)$, for some $d \in \mathbb{N}$, such that
    \begin{enumerate}
        \item $\sum_s \ket{w_s^x} = \ket{w}$
        \item $\braket{w_s^x}{w_t^y} = 0$ for $s \not\simH t$, $s \ne t$
        \item $\braket{w_s^x}{w_s^y} \le 0$ for $x \simG y$
        \item $\braket{w_s^x}{w_t^x} = 0$ for $s \ne t$
        \item $\braket{w_s^x}{w_t^y} \ge 0$ for $s \ne t$.
    \end{enumerate}
\end{theorem}

The proof is a straightforward modification of the proof for \cref{thm:th_homb}.
Before proceeding with this, it is necessary to express $\thmbar$ in a form analogous
to~\eqref{eq:th_max_sdp}.
This characterization appears without proof in~\cite{galtman2000spectral}; we give the
proof below.
We do not know how to provide such a formulation for $\thpbar$, so it may be possible that
$\thpbar$ does not admit an if-and-only-if statement along the lines of
\cref{thm:th_homb,thm:sch_mon_iff}.

\begin{theorem}
    \label{thm:thm_max_IT}
    \begin{align}
        \thmbar(G) &= \max\{ \opnorm{I+T} : I + T \succeq 0,
            \notag \\ &\hphantom{= \max\{\;}
            T_{ij} = 0 \textnormal{ for } i \not\sim j,
            \notag \\ &\hphantom{= \max\{\;}
            T_{ij} \ge 0 \textnormal{ for all } i,j \}.
            \label{eq:thm_IT}
    \end{align}
\end{theorem}
\begin{proof}
    The dual to the semidefinite program~\eqref{eq:thm_min_sdp} is~\cite{1056072}
    \begin{align}
        \thmbar(G) &= \max\{ \ip{B,J} : B \succeq 0,
        \notag \\ &\hphantom{= \max\{\;}
        \Tr B = 1,
        \notag \\ &\hphantom{= \max\{\;}
        B_{ij} = 0 \textrm{ for } i \not\sim j, i \ne j,
        \notag \\ &\hphantom{= \max\{\;}
        B_{ij} \ge 0 \textrm{ for all } i,j \}.
        \label{eq:thm_BJ}
    \end{align}
    Let $T$ be the optimal solution for~\eqref{eq:thm_IT}.
    We will show that this induces a feasible solution for~\eqref{eq:thm_BJ} via the recipe
    \begin{align*}
        B = \proj{\psi} \circ (I + T),
    \end{align*}
    where $\ket{\psi}$ is the eigenvector corresponding to the largest eigenvalue of $I+T$.
    This is positive semidefinite (being the Schur--Hadamard product of two positive semidefinite
    matrices), and $\ip{B,J} = \braopket{\psi}{I+T}{\psi} = \lambda$.
    $T_{ii}$ vanishes, so the diagonal of $B$ is equal to the diagonal of
    $\ket{\psi}\bra{\psi}$; consequently $\Tr B = 1$.
    The matrix $I+T$ has nonnegative entries so its eigenvector $\ket{\psi}$
    can be chosen nonnegative, leading to to $B_{ij} \ge 0$.
    So $B$ is feasible for~\eqref{eq:thm_BJ} and \eqref{eq:thm_BJ} $\ge$ \eqref{eq:thm_IT}.

    Conversely, suppose that $B$ is feasible for~\eqref{eq:thm_BJ} with value $\lambda$.
    Let $D$ be the diagonal component of $B$.  Let $D^{-1/2}$ be the diagonal matrix having
    entries $D_{ii} = 1/\sqrt{B_{ii}}$ with the convention $1/0=0$ (note that $D^{-1/2}$ is the
    Moore--Penrose pseudoinverse of $D^{1/2}$).
    Define
    \begin{align*}
        T = D^{-1/2} (B - D) D^{-1/2}.
    \end{align*}
    When $i \not\sim j$, this matrix satisfies $T_{ij}=0$.
    Since $D$ and $B-D$ have nonnegative entries, $T$ does as well.
    We have
    \begin{align}
       I+T &\succeq D^{-1/2} D D^{-1/2} + T \notag
        \\ &= D^{-1/2} B D^{-1/2} \label{eq:IT_vs_DBD}
        \\ &\succeq 0. \notag
    \end{align}
    So $T$ is feasible for~\eqref{eq:thm_IT}.
    Let $\ket{\psi}$ be the vector with coefficients $\psi_i = \sqrt{B_{ii}}$.
    Since $\Tr B = 1$, this is a unit vector.
    Making use of~\eqref{eq:IT_vs_DBD},
    \begin{align}
        \braopket{\psi}{I+T}{\psi} \notag
        &\ge \braopket{\psi}{D^{-1/2} B D^{-1/2}}{\psi}
        \\ &= \sum_{\substack{ij \textrm{ s.t.} \\ B_{ii}B_{jj} \ne 0}} B_{ij} \notag
        \\ &= \sum_{ij} B_{ij} \label{eq:B_off_diag}
        \\ &= \ip{J,B} = \lambda. \notag
    \end{align}
    Equality~\eqref{eq:B_off_diag} holds because $B$ is positive semidefinite and so
    satisfies $B_{ij}=0$ when $B_{ii} B_{jj} = 0$.
    Since $T$ is feasible for~\eqref{eq:thm_IT},
    \begin{align*}
        \eqref{eq:thm_IT} \ge \opnorm{I+T} \ge \lambda = \eqref{eq:thm_BJ}.
    \end{align*}
\end{proof}

\begin{proof}[Proof of \cref{thm:sch_mon_iff}]
    As in the proof of \cref{thm:th_homb}, we work with the Gram matrix of the $\ket{w_s^x}$
    vectors.  The existence of vectors satisfying the conditions in the theorem statement is
    easily seen to be equivalent to the existence of a matrix $C : \opV{G} \ot \opV{H}$
    satisfying
    \begin{align*}
        &\, C \succeq 0
        \\ &\, \sum_{st} C_{xyst} = 1
        \\ &\, C_{xyst} = 0 \textnormal{ for } s \not\sim t, s \ne t
        \\ &\, C_{xyss} \le 0 \textnormal{ for } x \sim y
        \\ &\, C_{xxst} = 0 \textnormal{ for } s \ne t
        \\ &\, C_{xyst} \ge 0 \textnormal{ for } s \ne t
    \end{align*}
    Using this characterization, we proceed with the proof.

    ($\Longrightarrow$):
    Suppose $\thmbar(G) \le \thmbar(H)$.  We will explicitly construct a matrix $C$ having the above
    properties.
    Let $\lambda = \thmbar(H)$.
    By \cref{thm:thm_max_IT} there is a matrix $T$ such that $\opnorm{I+T}=\lambda$,
    $I+T \succeq 0$,
    $T_{st}=0$ for $s \not\sim t$, and $T_{st} \ge 0$ for all $s,t$.
    Let $\ket{\psi}$ be the vector corresponding to the largest eigenvalue of $I+T$,
    which can be chosen nonnegative since $T$ is entrywise nonnegative.
    With $\circ$ denoting the Schur--Hadamard product, define the matrices
    \begin{align*}
        D &= \proj{\psi} \circ I,
        \\ B &= \proj{\psi} \circ (I + T).
    \end{align*}
    These are entrywise nonnegative.
    With $J$ being the all-ones matrix and $\ip{\cdot,\cdot}$ denoting the Hilbert--Schmidt inner
    product, it is readily verified that
    \begin{align*}
        \ip{D, J} &= \braket{\psi}{\psi} = 1, \\
        \ip{B, J} &= \braopket{\psi}{I+T}{\psi} = \lambda.
    \end{align*}
    Schur--Hadamard products between positive semidefinite matrices yield positive semidefinite
    matrices.
    As a consequence, $B \succeq 0$ and
    \begin{align*}
        \opnorm{I+T} = \lambda &\implies \lambda I - (I+T) \succeq 0 \implies \lambda D - B
        \succeq 0.
    \end{align*}

    Since $\lambda \ge \thmbar(G)$, there is a matrix $Z$ such that
    $Z \succeq 0$,
    $Z_{xx} = \lambda-1$ for all $x$, and
    $Z_{xy} \le -1$ for all $x \sim y$.
    Note that \eqref{eq:thm_min_sdp} gives existence of a matrix with $\thmbar(G)-1$ on the diagonal, but
    since $\lambda \ge \thmbar(G)$ we can add a multiple of the identity to get $\lambda-1$ on the
    diagonal.

    We now construct $C$.  Define
    \begin{align*}
        C = \lambda^{-1} \left[ J \ot B + (\lambda-1)^{-1} Z \ot (\lambda D - B) \right].
    \end{align*}
    Since $J$, $B$, $Z$, and $\lambda D - B$ are all positive semidefinite, and $\lambda-1 \ge 0$,
    we have that $C$ is positive semidefinite.
    The other desired conditions on $C$ are easy to verify.
    For all $x,y$ we have
    \begin{align*}
        \sum_{st} C_{xyst}
        &= \lambda^{-1} \left[ \ip{B,J} +
            (\lambda-1)^{-1} Z_{xy} [\lambda\ip{D, J}-\ip{B, J}] \right]
        \\ &= 1.
    \end{align*}
    For $s \not\sim t$, $s \ne t$, we have that $B_{st} = D_{st} = 0$ so
    $C_{xyst} = 0$.
    For $x \sim y$,
    \begin{align*}
        C_{xyss} &= \lambda^{-1} \left[ B_{ss} + (\lambda-1)^{-1} Z_{xy} (\lambda D_{ss} -
            B_{ss}) \right]
        \\ &= \lambda^{-1} \left[ D_{ss} + (\lambda-1)^{-1} Z_{xy} (\lambda D_{ss} - D_{ss}) \right]
        \\ &= \lambda^{-1} D_{ss} \left[ 1 + Z_{xy} \right]
        \\ &\le 0.
    \end{align*}
    For all $x$ and for $s \ne t$,
    \begin{align*}
        C_{xxst} &= \lambda^{-1} \left[ B_{st} + (\lambda-1)^{-1} Z_{xx} (\lambda D_{st} -
            B_{st}) \right]
        \\ &= \lambda^{-1} B_{st} \left[ 1 - (\lambda-1)^{-1} Z_{xx} \right]
        \\ &= 0.
    \end{align*}
    For all $x,y$ and for $s \ne t$,
    \begin{align*}
        C_{xyst} &= \lambda^{-1} \left[ B_{st} + (\lambda-1)^{-1} Z_{xy} (\lambda D_{st} -
            B_{st}) \right]
        \\ &= \lambda^{-1} B_{st} \left[ 1 - (\lambda-1)^{-1} Z_{xy} \right]
        \\ &\ge 0,
    \end{align*}
    where the last inequality follows from the fact that
    $Z \succeq 0 \implies \abs{Z_{xy}} \le \max\{ Z_{xx}, Z_{yy} \} = \lambda-1$.

    ($\Longleftarrow$):
    Let $Z$ achieve the optimal value (call it $\lambda$) for the minimization
    program~\eqref{eq:thm_min_sdp} for $\thmbar(H)$.
    We will provide a feasible solution for~\eqref{eq:thm_min_sdp} for $\thmbar(G)$ to show that
    $\thmbar(G) \le \thmbar(H)$.
    Specifically, let 
    \begin{align*}
    Y = \big(I \otimes \bra{\mathbf{1}}\big) \big[\big(J \otimes Z\big) \circ C\big] \big(I \otimes \ket{\mathbf{1}} \big),
    \end{align*}
    as in the proof of \cref{thm:th_homb}.
    Since $C \succeq 0$ and $Z \succeq 0$, and positive semidefiniteness is preserved by conjugation, we have that $Y \succeq 0$. Considering the entries of $Y$ we see that
    $Y_{xy} = \sum_{st} Z_{st} C_{xyst}$.

    Using the fact that $Z_{ss} = \lambda-1$ and $C_{xxst}=0$ for $s \ne t$, we have
    \begin{align*}
        Y_{xx} &= \sum_{st} Z_{st} C_{xxst}
        = (\lambda-1) \sum_{st} C_{xxst} = \lambda-1.
    \end{align*}
    For $x \sim y$ we have
    \begin{align*}
        Y_{xy} &= \sum_{st} Z_{st} C_{xyst}
        \\ &= \sum_{s \sim t} \underbrace{Z_{st}}_{\le -1}
            \underbrace{C_{xyst}}_{\ge 0} +
            \sum_{s \not\sim t, s \ne t} Z_{st} \underbrace{C_{xyst}}_{=0} +
            \sum_{s} \underbrace{Z_{ss}}_{\ge -1} \underbrace{C_{xyss}}_{\le 0}
        \\ &\le \sum_{s \sim t} (-1) C_{xyst} +
            \sum_{s \not\sim t, s \ne t} (-1) C_{xyst} +
            \sum_{s} (-1) C_{xyss}
        \\ &= \sum_{st} (-1) C_{xyst} = -1.
    \end{align*}
    Now define a matrix $Y'$ consisting of the real part of $Y$ (i.e.\ with coefficients
    $Y_{xy}' = \textrm{Re}[Y_{xy}]$).
    This matrix is real, positive semidefinite, and satisfies $Y_{xx}=\lambda-1$ for
    all $x$ and $Y_{xy} \le -1$ for $x \sim y$.
    Therefore $Y'$ is feasible for~\eqref{eq:thm_min_sdp} with value $\lambda=\thmbar(H)$.
    Since $\thmbar(G)$ is the minimum possible value of~\eqref{eq:thm_min_sdp}, we have
    $\thmbar(G) \le \thmbar(H)$.
\end{proof}

By setting $G=K_n$ or $H=K_n$ it is possible to formulate corollaries analogous to
\cref{thm:beigi_beta,thm:beigi_beta_chi}.
We describe only the first of these here.

\begin{corollary}
    Let $\beta^-(H)$ be the largest $n$ such that there are vectors
    $\ket{w} \ne 0$ and $\ket{w_s^x} \in
    \mathbb{C}^d$ for each $x \in \{1,\dots,n\}$,
    $s \in V(H)$, for some $d \in \mathbb{N}$, such that
    \begin{enumerate}
        \item $\sum_s \ket{w_s^x} = \ket{w}$
        \item $\braket{w_s^x}{w_t^y} = 0$ for $s \simH t$
        \item $\braket{w_s^x}{w_s^y} \le 0$ for $x \ne y$
        \item $\braket{w_s^x}{w_t^x} = 0$ for $s \ne t$
        \item $\braket{w_s^x}{w_t^y} \ge 0$ for $s \ne t$.
    \end{enumerate}
    Then $\beta^-(H) = \lfloor \thm(H) \rfloor$.
\end{corollary}

\begin{IEEEbiographynophoto}{Toby~Cubitt}
    received his Bachelor's and Master's degrees from the
    University of Cambridge, U.K., in 2002. He carried out his doctoral
    studies at the Max Planck Institute for Quantum Optics, Garching,
    Germany, and received his Ph.D. degree from the Technische Universit\"at
    M\"unchen in 2006. He was a Research Assistant and then a Leverhulme
    Early Career Fellow at Bristol University until 2010, a Juan de la Cierva
    Fellow at the Universidad Complutense de Madrid until 2012, and now holds
    a Royal Society University Research Fellowship in the Department for
    Applied Mathematics and Theoretical Physics at the University of
    Cambridge, U.K.
\end{IEEEbiographynophoto}

\begin{IEEEbiographynophoto}{Laura~\texorpdfstring{Man\v{c}inska}{Mancinska}}
    is currently a research fellow in Stephanie Wehner's group at the
    Centre for Quantum Technologies, National University of Singapore. She received her
    PhD from the Combinatorics and Optimization department of the University of Waterloo,
    Canada where she studied under the guidance of Andrew Childs and Debbie Leung. In her
    PhD thesis she investigated a restricted model of quantum operations in which
    separated parties are only allowed to perform local quantum maps and communicate
    classically (LOCC). Her research interests also include other topics in quantum
    information such as nonlocal games, entanglement assisted zero-error communication
    etc.
\end{IEEEbiographynophoto}

\begin{IEEEbiographynophoto}{David~Roberson}
    is currently a research fellow under Dmitrii Pasechnik at Nanyang Technological
    University. He received his PhD from the department of Combinatorics \& Optimization at
    the University of Waterloo, Canada where he studied under the guidance of Chris
    Godsil. In his PhD thesis he introduced and investigated quantum homomorphisms, a
    class of non-local games which generalize quantum colorings. His research interests
    also include zero-error communication, graph parameters, and cores of graphs.
\end{IEEEbiographynophoto}

\begin{IEEEbiographynophoto}{Simone~Severini}
    is a Reader (Associate Professor) in Physics of Information and a Royal Society
    University Research Fellow at UCL\@. Previously, he was a Newton International Fellow at
    UCL, a Postdoctoral Research Associate in the Institute for Quantum Computing at the
    University of Waterloo, and in the Department of Mathematics at the University of
    York. A visiting student at UC Berkeley, he obtained his PhD from Bristol University,
    under the supervision of Richard Jozsa. Among other institutions, he has been a long
    term visiting scientist at the Cesarea Rothschild Institute, MIT, and Shanghai
    Jiao-Tong University. His interests are at the interface between discrete mathematics,
    physics, and complex systems.
\end{IEEEbiographynophoto}

\begin{IEEEbiographynophoto}{Dan~Stahlke}
    received an MS in Physics from the University of Alaska Fairbanks, in 2010, studying
    nonlinear dynamics under the supervision of Renate Wackerbauer, and a PhD in Physics
    from Carnegie Mellon University, in 2014, studying quantum information under the
    supervision of Robert B. Griffiths.
    From 2002-2010 he was a computer programmer at the Geographic Information Network of
    Alaska, University of Alaska Fairbanks.
    He will soon be a Rotation Engineer at Intel Corporation.
\end{IEEEbiographynophoto}

\begin{IEEEbiographynophoto}{Andreas~Winter}
    received a Diploma degree in Mathematics from the Freie Universit\"{a}t Berlin,
    Berlin, Germany, in 1997, and a Ph.D. degree from the Fakult\"{a}t f\"{u}r Mathematik,
    Universit\"{a}t Bielefeld, Bielefeld, Germany, in 1999.  He was Research Associate at
    the University of Bielefeld until 2001, and then with the Department of Computer
    Science at the University of Bristol, Bristol, UK\@. In 2003, still with the
    University of Bristol, he was appointed Lecturer in Mathematics, and in 2006 Professor
    of Physics of Information.  Since 2012 he is ICREA Research Professor with the
    Universitat Aut\`{o}noma de Barcelona, Barcelona, Spain.
\end{IEEEbiographynophoto}

\end{document}